\documentclass[11pt]{article}
\usepackage{odonnell}

\newcommand{\ketbra}[2]{\ket{#1}\!\bra{#2}}

\newcommand{\be}{\begin{equation}}
\newcommand{\ee}{\end{equation}}

\newcommand{\rmd}{{\mathrm d}}

\newcommand{\lovasz}{Lov{\'a}sz\xspace}
\renewcommand{\th}{\vartheta}

\newcommand{\G}{K\!G}
\newcommand{\Geven}{\G_{\mathrm{even}}}
\newcommand{\Gekr}[1]{\G_{\mathrm{EKR}_{#1}}}

\newcommand{\J}{{\mathbbm J}}

\let\origi\i
\renewcommand{\i}{\mathrm{i}}

\newcommand{\X}{\chi}

\newcommand{\ulX}{\X}

\newcommand{\p}{h}
\newcommand{\boldp}{\bh}
\renewcommand{\N}{n}
\newcommand{\Nbig}{N}
\newcommand{\NN}{\mathbbm N}

\newcommand{\Sig}{\Sigma}
\newcommand{\bSig}{\bSigma}

\newcommand{\CC}{\mathfrak{C}}

\newcommand{\OurPsi}{\Psi}

\newcommand{\Opt}{\mathrm{Opt}}
\newcommand{\Optgauss}{\mathrm{Opt}_{\mathrm{Gauss}}}
\newcommand{\SDPgauss}{\mathrm{SDP}_{4,\mathrm{Gauss}}}
\newcommand{\SOS}{\mathrm{SOS}}
\newcommand{\SOSf}{\mathrm{SOS_{4,2}}}
\newcommand{\SYK}{\mathrm{SYK}}
\newcommand{\SYKtwo}{\mathrm{SYK}^{\textnormal{2col}}}

\newcommand{\pE}{\wt{\E}}

\newcommand{\Ad}{\mathrm{Ad}}

\renewcommand{\Vec}{\mathrm{Vec}}
\newcommand{\Mat}{\mathrm{Mat}}

\begin{document}

\title{Optimizing Strongly Interacting Fermionic Hamiltonians}

\author{Matthew B. Hastings\thanks{Station Q.\ and Microsoft Quantum. \texttt{mahastin@microsoft.com}}
\and
Ryan O'Donnell%
\thanks{Microsoft Quantum and Carnegie Mellon University Computer Science Department. \texttt{odonnell@cs.cmu.edu}}
}

\date{\today}

\maketitle

\begin{abstract}
The fundamental problem in much of physics and quantum chemistry is to optimize a low-degree polynomial in certain anticommuting variables.
Being a quantum mechanical problem, in many cases we do not know an efficient classical witness to the optimum, or even to an approximation of the optimum.
One prominent exception is when the optimum is described by a so-called ``Gaussian state", also called a free fermion state.  In this work we are interested in the complexity of this optimization problem when no good Gaussian state exists.  Our primary testbed is the Sachdev--Ye--Kitaev (SYK) model of random degree-$q$ polynomials, a model of great current interest in condensed matter physics and string theory, and one which has remarkable properties from a computational complexity standpoint.  Among other results, we give an efficient classical certification algorithm for upper-bounding the largest eigenvalue in the $q=4$ SYK model, and an efficient quantum certification algorithm for lower-bounding this largest eigenvalue; both algorithms achieve constant-factor approximations with high probability.
\end{abstract}

\section{Introduction}

\subsection{Classical optimization}
Classical constraint satisfaction problems (CSPs) are very well studied in theoretical computer science, from the perspectives of both worst-case complexity and average-case complexity.
We know many ingenious approximation algorithms --- often based on semidefinite programming (SDP)~\cite{RS09} --- for efficiently finding solutions within a constant factor of optimal.
At the same time, we also have very precise predictions for the computational intractability of achieving certain approximation factors, subject to conjectures like $\PTIME \neq \NP$~\cite{Has01}, the Unique Games Conjecture~\cite{Rag09}, and the optimality of Sum-of-Squares (SOS) hierarchies~\cite{KMOW17}.
CSPs involve finding an assignment (typically a Boolean one) to a large number~$n$ of variables so as to satisfy many given constraints, each a predicate involving at most some~$q$ variables.
As such, they can often be expressed as finding the maximum value of an $n$-variate degree-$q$ polynomial over assignments from $\{\pm 1\}^n$; in other words (and in the homogeneous case),
\begin{equation} \label{CSP}
    \max\Bigl\{\sum_{\substack{S \subseteq [n] \\ |S| = q}} a_S \prod_{j \in S} \X_j  \quad : \quad \X_j^2 = 1\ \forall j\Bigr\}.
\end{equation}
For example, the $q = 2$ case is essentially the ``Max-Cut'' problem (with edge-weights), and the famous Goemans--Williamson SDP-rounding algorithm and its successors~\cite{GW95,CW04} can efficiently find an $O(\log(1/\eps))$-factor approximation to the optimum in any worst-case instance when the optimum is at least $\eps |a|_1$.
Regarding $q = 2$ in the average case, if the $a_S$'s are chosen to be independent standard Gaussian random variables, then we have the Sherrington--Kirkpatrick model, and Montanari~\cite{Mon21} has given an efficient algorithm for finding a $(1-o(1))$-factor approximation to the optimum, with high probability.

On the other hand (and oversimplifying the story), CSP optimization tends to be far more intractable once~$q \geq 3$.
For example, work of H{\aa}stad and successors~\cite{Has01,MR10}, shows that there are worst-case $q=3$ instances of \eqref{CSP}, with $m = O(n)$ nonzero coefficients $a_S \in \{\pm 1\}$, where the optimum is $(1-\eps)|a|_1 = (1-\eps)m$, but it is $\NP$-hard to find a solution of value at least~$m/f(n)$ for some slowly growing $f(n) \to \infty$ (and likely this cannot even be done in time $2^{n^{.99}}$).
Furthermore, for \emph{random} instances  of the above form, even though the optimum is typically $O(\sqrt{n})$ we have significant evidence (in the form of SOS hierarchy lower bounds) that any algorithm for certifying the optimum is at most $.99m$ requires~$2^{\Omega(n)}$ time~\cite{Gri01}.

\ignore{Semidefinite relaxations of classical optimizations problems are well-studied.  The famous Goemans-Williamson algorithm CITE is now recognized as the first level of a hierarchy of polynomial time classical algorithms CITE.  Under the unique-games conjecture, these algorithms achieve the optimal approximation ratio for certain optimization problems CITE.
}

\subsection{Fermionic optimization}
There exist analogous optimization problems in the quantum setting.
Here we must distinguish two cases.
In one case, called a ``spin system", the Hilbert space of the quantum system has a tensor product structure, and the analogue of~\eqref{CSP} is called the \emph{$q$-Local Hamiltonian problem}.
In the other case, one considers fermionic degrees of freedom, which obey the so-called canonical anticommutation relations.\footnote{One can of course consider systems with both fermionic and spin degrees of freedom.}
This second case is the focus of our paper.
It is well studied, both from a mathematical physics point of view and for applications in quantum chemistry, condensed matter physics, and materials science.

In addition to being describable in physics terms, a fermionic optimization can be posed rather similarly to~\eqref{CSP}, but with the following differences:
First, the indeterminates $\X_j$ are to be assigned \emph{self-adjoint matrices}, rather than numbers, and it is required that they pairwise \emph{anticommute}.
Such matrices are known as Majorana operators.
Second, the quantity to be maximized is the largest \emph{eigenvalue} of the resulting matrix.
Thus a typical fermionic optimization problem may be stated as follows:
\begin{equation} \label{FCSP}
    \max\Bigl\{\lambda_{\max}(h),\ h = \sum_{\substack{S = \{j_1, \dots, j_q\} \subseteq [n] \\ j_1 < \cdots < j_q}} a_S \X_{j_1} \cdots \X_{j_q}  \quad : \quad
    \X_j^* = \X_j\ \forall j,\ \X_j^2 = \Id\ \forall j,\ \X_j \X_k = -\X_k \X_j\ \forall j,k\Bigr \},
\end{equation}
where $\Id$ denotes the identity operator.\footnote{We remark that this problem is only physically natural if $q$ is even.
One should also arrange for $h$ to be self-adjoint and hence have real eigenvalues; this occurs, e.g., if $q$ is a multiple of $4$ and each $a_S$ is real.}

Although \eqref{FCSP} as phrased involves searching over all valid matrix-assignments for the~$\X_j$'s, this is not the standard viewpoint in physics.
Rather (as is reviewed in \Cref{sec:represents}), it is without loss of generality to fix a particular canonical assignment of $D$-dimensional ``gamma matrices'' $\gamma_j$ to the~$\X_j$'s, where $D = 2^{n/2}$.
Thus~$h$, termed the ``Hamiltonian'' in physics, may be thought of as an implicitly-represented matrix of exponentially large dimension, and the only task is to determine its largest eigenvalue.
To recall some further physics terminology, the expectation value of the matrix~$h$ in some pure quantum state~$\ket{\psi}$ --- i.e., the value of vector $\ket{\psi}$ under~$h$'s quadratic form --- is called its ``energy''.
Usually in physics one is concerned with the smallest possible energy, termed the ground state energy, but we will follow the computer science convention of seeking the maximum eigenvalue.
(Of course, these are equivalent upon replacing~$h$ by~$-h$.)%

The particular case of $q = 4$ corresponds to finding the binding energy of a molecule with $2$-body interactions, and is therefore one of the most basic computational tasks in quantum chemistry.  Here the binding energy is the difference between the ground state energy for some given molecule and the sum of ground state energies of its constituent atoms.
\footnote{A remark regarding applications in quantum chemistry: in many cases one can define a ``particle number operator'' $\i\sum_{j=1}^{n/2} \X_{2j-1} \X_{2j}$ which commutes with the Hamiltonian~$h$.  In this case, one may replace the Majorana operators with certain linear combinations called creation and annihilation operators so that every term in the Hamiltonian has an equal number of creation and annihilation operators.  For practical purposes, this representation with creation and annihilation operators is convenient, as for such Hamiltonians one immediately knows that the expectation value of an operator that does not commute with the particle number vanishes.  However, in this paper we will not focus on this case.  Interestingly, though, several of the constraints considered in quantum chemistry can be more readily expressed in terms of the Majorana operators. For example, the so-called ``$T_1$ constraint" is a constraint on expectation values of degree-$4$ polynomials that follows from positivity of the ``degree-$6$ pseudoexpectation'' value, and (in the Majorana language) from the fact that a degree-$6$ polynomial in Majorana operators with real coefficients is skew-adjoint.}

\subsection{Computational complexity}   \label{sec:cc}
The $q = 2$ case of fermionic optimization~\eqref{FCSP} can be solved by an efficient classical algorithm, involving diagonalizing the Hamiltonian.
Moreover, the algorithm can  find an implicit representation of the optimal state, which will be a so-called \emph{(fermionic) Gaussian state}~\cite{TD02,Bra05}.
However, for $q > 2$ the problem becomes computationally difficult.
To see this, one may observe that the collection of operators $\i \X_{2j-1} \X_{2j}$ (for $j = 1 \dots n/2$) square to the identity and also \emph{commute}.
It follows from this that one can (efficiently) encode any $n$-variable, degree-$q$ instance of a Boolean CSP (as in~\eqref{CSP}) by a $2n$-variable, degree-$2q$ instance of the fermionic problem~\eqref{FCSP}.
(Furthermore, in this encoding, commuting Boolean assignments correspond to fermionic Gaussian states.)
This reduction shows that the degree-$4$ fermionic problem is $\NP$-hard.
In fact, one may also encode the $q$-Local Hamiltonian problem by a degree-$2q$ fermionic problem~\cite{LCV07}, and it therefore follows from~\cite{KKR06} that the degree-$4$ fermionic optimization problem is even $\QMA$-hard.

That said, there are quite good approximation algorithms for degree-$2$ Boolean constraint satisfaction (via Goemans--Williamson~\cite{GW95} and its extensions), whereas for degree-$3$ and above there are strong inapproximability results.
The latter observation shows that degree-$6$ (and above) fermionic optimization will be hard to approximate well, but leaves open the possible that the $q=4$ case (which is most relevant for quantum chemistry) may be a ``sweet spot'' in terms of complexity.

In this work we focus on such approximation problems.
In the case of quantum optimization problems, though, a fundamentally new difficulty arises.
In the classical case of~\eqref{CSP}, one can naturally certify that the optimum is at least some~$\alpha$ by providing a witnessing assignment $\chi \in \{\pm 1\}^n$.
By contrast, in the quantum case, there may well be no efficient way to classically represent a witness (such as a state~$\ket{\psi}$) to the ground state energy, or even to an energy reasonably close to the optimum.
Indeed, assuming $\mathsf{QMA} \neq \mathsf{NP}$, this would be implied by the \emph{quantum PCP conjecture}~\cite{AAV13} (at least for local tensor-product Hamiltonians), which posits that determining the ground state energy up to a constant factor
is $\mathsf{QMA}$-hard.  A noncomputational implication of the quantum PCP conjecture is the \emph{NLTS (no low-energy trivial states) conjecture}~\cite{FH14}, which
posits the nonexistence of a certain class of easily described states with energy close to optimum.
It is natural to define analogous conjectures, fermionic quantum PCP (fqPCP) and fermionic NLTS (fNLTS), for sparse instances of fermionic optimization problems, but this seems not to have been studied.

Besides trying to certify that the optimum of~\eqref{FCSP} is at least some~$\alpha$, one may also try to certify that it is at \emph{most} some~$\beta$.
For this task, it is natural to consider SDP relaxations (as was done in~\cite{BH13} for the Local Hamiltonian problem), and in particular SOS hierarchies.
Most generally one may use the non-commutative SOS hierarchy, developed in, e.g.,~\cite{HM04,NPA08,DLTW08,PNA10}; it provides increasingly tight upper bounds $\beta_q \geq \beta_{q+2} \geq \beta_{q+4} \geq \cdots$ on~\eqref{FCSP}, with~$\beta_k$ being certifiable in deterministic classical $n^{O(k)}$ time.
In the specific context of fermionic optimization, a closely related SDP hierarchy has been studied as a way of characterizing which two-particle reduced density matrices can arise from a many-body system (i.e., which collections of degree-$4$ expectation values $\tr(\rho \gamma_i \gamma_j \gamma_k \gamma_l)$ can arise from a genuine state~$\rho$).
See, e.g.,~\cite{coleman1963structure,erdahl1978representability,percus1978role,mazziotti2001uncertainty,Maz12,klyachko2006quantum}.
This so-called ``$N$-representability problem" is $\QMA$-hard, and is perhaps the fundamental problem of quantum chemistry~\cite{LCV07}.

\subsection{Strongly interacting Hamiltonians, and SYK model}
As mentioned, in the case $q=2$, the fermionic optimization problem~\eqref{FCSP} is solvable efficiently, with the optimizer being an efficiently-representable ``Gaussian state''.
Given this fact, previous research (e.g.~\cite{BGKT19}) has studied how close Gaussian states can come to achieving optimality in the $q = 4$ case of~\eqref{FCSP}, and how efficiently one can optimize over Gaussian states in this case.
In the present work, our main interest is the complexity of optimizing fermionic Hamiltonians  when \emph{no} good Gaussian state exists.
A prominent example of this, and our main testbed, is the ``SYK model'' of random fermionic Hamiltonians.\\

The \emph{Sachdev--Ye--Kitaev (SYK) model} refers to the case of~\eqref{FCSP} where the coefficients~$a_S$ are chosen to be independent random Gaussians (and where $q$~and~$n$ are even).%
\footnote{In all cases, a prefactor of $\i^{\binom{q}{2}}$ should be included to ensure self-adjointness.
We remark that although the case of odd  $q$ or~$n$ is not natural physically, these cases are still mathematically interesting.}
Contrary to the traditional physics normalization, we will scale by $1/\sqrt{\binom{n}{q}}$ so that $|a|_2^2 = \sum_S a_S^2 = 1$ in expectation.
The SYK model has attracted considerable interest in physics for being ``maximally chaotic", and for having a ground state that is, in a sense, as far from being a Gaussian state as is possible for a low-degree Hamiltonian.
With study dating back to~\cite{FW70,BF71}, the canonical references for this model are~\cite{SY93,Kit15}.
A review from a high energy physics perspective is~\cite{rosenhaus2019introduction}.
The physics papers~\cite{GV16,GJV18} are rather accessible for those with a mathematics background, and the works~\cite{FTW19,FTW18,FTW20} give some entirely mathematical results.

Considerations from mathematical physics suggest that for~$h$ drawn from the degree-$4$ SYK model, the maximum value in~\eqref{FCSP} --- which we will denote by $\Opt(h)$ --- is $\Theta(\sqrt{n})$\footnote{We remark that under the conventional normalization for SYK used in physics, this corresponds to ground state energy~$-\Theta(n)$.}
with high probability.
Indeed, as we discuss in \Cref{sec:heur}, even finer-grained details are predicted, including that the leading constant hidden in the~$\Theta(\cdot)$ should be~$\frac{1}{2\sqrt{2}}$.
On the other hand, the only mathematical proofs we are aware of come from~\cite{FTW19}, from which it follows that $\omega(1) < \Opt(h) \leq \sqrt{\ln 2} \sqrt{n}$ with high probability.

\subsection{Summary of our results}

We provide a variety of new results concerning efficient optimization of fermionic Hamiltonians.
Primary among these are a rigorous proof that $\Opt(h) = \Theta(\sqrt{n})$ (with high probability) for $h$ drawn from the degree-$4$ SYK model.
Moreover, we prove this via efficient \emph{lower and upper bound certification algorithms}.
For the upper bound, we show that degree-$6$ SOS works, and we prove a tighter bound on the numerical constant for degree-$8$ SOS:
\begin{theorem}                                     \label{thm:main1}
    The degree-$8$ SOS relaxation, which is solvable in deterministic $\poly(n)$ time,
    certifies that $\Opt(h) \leq \sqrt{1+\sqrt{6}}\sqrt{n}$ with high probability over $h$ drawn from the degree-$4$ SYK model.
    (Moreover, even the degree-$6$ SOS relaxation certifies that $\Opt(h)\leq O(\sqrt{n})$ with high probability.)
\end{theorem}
\begin{remark}
    We also show (\Cref{thm:fool}) that degree-$4$ SOS \emph{fails} to certify anything better than $\Opt(\bh) \lesssim n$.
    On the other hand, our certification procedure can be carried out in a natural ``fragment'' of degree-$6$ SOS of size~$O(n^4)$.
    It would be interesting to see if the ``fragment'' introduced here can have practical applications in, e.g., the  simulation of quantum chemistry as in~\cite{Maz12}.
\end{remark}

Our lower bound, perhaps surprisingly, is via an efficient \emph{quantum} certification algorithm that creates a variational state of large expectation value.
It is an outstanding open problem as to whether there is an efficient classical certification algorithm (recall that
Gaussian states are known to only achieve expectation value~$O(1)$).
\begin{theorem}                                     \label{thm:main2}
    There is a $\poly(n)$-time quantum algorithm that, given any degree-$4$ Hamiltonian~$h$, returns a quantum state~$\rho$.
    With high probability over $h$ drawn from the degree-$4$ SYK model, the output state~$\rho$ has expectation value~$\Omega(\sqrt{\N})$.
\end{theorem}
\begin{remark}
    As we will discuss, this guarantee \emph{should} be thought of as providing a (quantumly-checkable) ``certificate'', since for \emph{every} $h$~and~$\rho$, there is a very simple $\poly(n)\cdot \log(1/\delta)$-time quantum algorithm that estimates the expectation value of~$h$ under~$\rho$ to any $\pm 1/\poly(n)$ accuracy with confidence~$1-\delta$.
\end{remark}

Together, \Cref{thm:main1,thm:main2} naturally suggest:
\paragraph{Open problem.} \emph{Are there efficient quantum/classical algorithms that certify the value of degree-$4$ SYK instances up to a $(1\pm o(1))$-factor (from above or below or both)?}

\paragraph{Open problem.} \emph{Are there sparse instances of \eqref{FCSP} with $q=O(1)$ that fulfill the fermionic NLTS conjecture}\footnote{The NLTS conjecture posits that there exists a family of quantum Hamiltonians which are ``spin systems", in that the Hilbert space has a local tensor product structure, where all terms in the Hamiltonian $h$ have norm $O(1)$, where each spin participates in $O(1)$ terms, and where each term acts on $O(1)$ spins, so that for some sufficiently constant $c>0$, there are no quantum states $\psi$ which are ``low energy" in that the difference between $\langle \psi|h |\psi\rangle$ and $\Opt(h)$ is $\leq c n$, and which are ``trivial" in that $\psi$ is given by acting on some product state by a bounded depth quantum circuit.  A fermionic analogue of this conjecture would postulate the same property for a system with fermionic degrees of freedom.}\\

We now discuss some additional results we prove.
First, we give the below nontrivial upper bound for \emph{worst-case} degree-$4$ Hamiltonians.
This theorem is used in establishing \Cref{thm:main1}.
\begin{theorem}                                     \label{thm:main3}
    Let $h$ be as in~\eqref{FCSP} with $q=4$, and assume for normalization that $\sum_S a_S^2 = 1$.
    Then
    \[
        \Opt(h) \leq \sqrt{\binom{\N/2}{2}} \sim \tfrac{1}{2\sqrt{2}} \N,
    \]
    with equality being attained by a certain Hamiltonian (essentially the square of a quadratic) in which all terms commute.
    (This may be contrasted with the commutative case of~\eqref{CSP}, where $\Opt(h) = \sqrt{\binom{n}{4}} \sim \tfrac{1}{\sqrt{24}} n^2$ is attainable.)
\end{theorem}

The upper bound in this theorem relies only  on the following algebraic relation: two monomials $\chi_{j_1}\chi_{j_2}\chi_{j_3}\chi_{j_4}$ and $\chi_{k_1}\chi_{k_2}\chi_{k_3}\chi_{k_4}$ anticommute if and only if $|\{j_1, \dots, j_4\} \cap \{k_1, \dots, k_4\}|$ is odd.
Indeed, we prove \Cref{thm:main3} by studying the optimization problem~\eqref{FCSP} in a more general context where one doesn't assume that $\chi_j$~and~$\chi_k$ anticommute for all $j \neq k$, but rather one assumes --- for some given graph $\Gamma = ([n], E)$ --- that they anticommute if $\{j,k\} \in E$ and they commute if $\{j,k\} \not \in E$.
In this more general setting, we show that for $h$ of degree $1$ with $\sum_j a_j^2 = 1$, it holds that $\alpha(\Gamma) \leq \Opt(h)^2 \leq \th(\Gamma)$, where $\alpha(\Gamma)$ denotes the independence number of~$\Gamma$, and $\th(\Gamma)$ denotes its \emph{\lovasz Theta Function} value.
Our \Cref{thm:main3} is then obtained by determining the \lovasz Theta Function of an appropriate Kneser-like graph family on vertex set~$\binom{[n]}{4}$.
The analogous problem for $\binom{[n]}{q}$, $q > 4$, remains open, as does the question of whether there exists~$h$ with $\Opt(h)^2 > \alpha(\Gamma)$ strictly.

From \Cref{thm:main1,thm:main3} we infer that for degree-$4$ $h$, the optimum value behaves very differently for typical (SYK) instances, versus worst-cases instances.
It would seem that $\Opt(h) \approx n$ is possible only if $h$ is (roughly) a sum of a small number of squared quadratics, in which case Gaussian states achieve a near-optimal expectation value (see \Cref{lowrank}).
We prove some results along these lines, including the following:
\begin{theorem}                                     \label{thm:main4}
     There is a $\poly(n)$-time classical algorithm with the following property:
     Given $h$ as in~\eqref{FCSP} with $q=4$, $\sum_S a_S^2 = 1$, and the property that the optimum Gaussian state achieves expectation value at least $\eps n$, the algorithm finds a Gaussian state achieving expectation value $\Omega(\frac{\eps}{\log(1/\eps)}) n$.
\end{theorem}
This result generalizes the earlier work of Bravyi, Gosset, K{\"o}nig, and Temme~\cite{BGKT19} which achieved an $\Omega(\frac{1}{\log n})$-approximation factor.

\section{Preliminaries}

\subsection{Notation}

\begin{notation}
    We use \textbf{boldface} to denote random variables.  We use $\SYK_q(\N)$ to denote the Sachdev-Ye-Kitaev model with a given choice of $q,\N$.  We use $\sim$ to indicate that a random variable is drawn from a particular distribution, such as
    $\bh \sim \SYK_4(\N)$.
\end{notation}
\begin{notation}
    For a natural number $\N$ we write $[\N] = \{1, 2, \dots, \N\}$.
\end{notation}
\begin{notation}
    For a parameter $\N \to \infty$, we sometimes write $\poly(\N) = \N^{O(1)} = O(\N^c)$ for an unspecified constant~$c$.
\end{notation}
\begin{notation}
    $\Id$ denotes the identity operator/matrix (where the dimension can always be inferred from context).
\end{notation}
\begin{notation}
    For an operator~$A$ of dimension $D$ we write $\Tr(A)$ for the trace of~$A$, and $\tr(A) = \frac{1}{D} \Tr(A)$ for the normalized trace.
\end{notation}
\begin{notation}
    If $G = (V,E)$ is a simple undirected graph, then $\overline{G}$ denotes the edge-complement graph~$(V,\binom{V}{2} \setminus E)$.
    Also, $K_\N$ denotes the complete graph on~$\N$ vertices, and hence $\ol{K}_\N$ denotes the empty graph on~$\N$ vertices.
\end{notation}
\begin{notation}
We use $|\cdot|_p$ to denote the $\ell_p$ norm of a vector; if $p$ is not written then this is the $\ell_2$ norm by default.
\end{notation}
\begin{notation}
$\| \cdot \|_{\mathrm{op}}$ denotes the operator norm of a matrix.
However $\Vert \cdot \Vert_p$ denotes the Schatten $p$-norm of a matrix.
\end{notation}
\begin{notation}
We use both ket notation, such as $\ket{\phi}$, and vector notation, such as $\vec a$, to denote vectors.  Generally we try to use ket notation for pure quantum states and vector notation for everything else, but we make some exceptions.  For example, if $\vec a$ is a vector, $\ketbra{a}{a}$ is used as convenient notation for the projector onto $\vec a$.
\end{notation}
\begin{notation}
We use $*$ to denote Hermitian conjugation of operators as well as to denote
a formal adjoint operator on polynomials.
\end{notation}

\subsection{Quasi-Clifford algebras}
We work in the algebra $\C^*\la\ulX\ra$ of polynomials with complex coefficients over a finite sequence $\ulX = (\X_1, \dots, \X_\N)$ of noncommuting indeterminates.
The empty word is denoted by~$1$.
There is a formal adjoint operator~$*$ on polynomials, and we will assume that all indeterminates $\X_j$ are self-adjoint, $\X_j^* = \X_j$.
This $*$~maps coefficients to their complex conjugates, and reverses words.

We furthermore introduce some relations.
First, we assume each $\X_j$ is an \emph{involution}; i.e., $\X_j^2 = 1$ for all~$j$.
Next, given an undirected graph $\Gamma = ([\N],E)$ on vertex set~$[\N]$, we introduce the relation $\X_j \X_k = -\X_k \X_j$ for all edges $\{j,k\} \in E$, and the relation $\X_j \X_k = \X_k \X_j$ for all nonedges $\{j,k\} \not \in E$.
We denote the resulting algebra by $\CC(\Gamma)$; it is (a special case of) a \emph{quasi-Clifford algebra}, as introduced by Gastineau-Hills~\cite{GH80,GH82}.
We will be particularly interested in the complete graph $\Gamma = K_{\N}$ with $\N$~even; then $\CC(K_{\N})$ is the ``usual'' Clifford algebra in which every pair of distinct indeterminates anticommutes.

Given our relations, it is easy to see that every monomial (product of indeterminates) is equivalent, up to sign change, to a multilinear monomial in which the indeterminates are ordered lexicographically.
Thus we can write a general polynomial $\p \in \CC(\Gamma)$ as
\begin{equation} \label{eqn:h}
    \p = \p(\X) = \sum_{S \subseteq [\N]} a_S \X^S, \quad \text{where }
    \X^S \coloneqq \X_{k_1} \X_{k_2} \cdots \X_{k_q} \text{ if } S = \{k_1, k_2, \dots, k_q\} \text{ with } k_1 < k_2 < \cdots < k_q,
\end{equation}
with coefficients $a_S \in \C$.%
\footnote{Notational notes: When discussing physics applications, we will typically write $J_S$ instead of~$a_S$.
Also, $\X^\emptyset$ stands for the empty monomial~$1$, and for $a \in \C$ we write just~$a$ instead of $a \cdot 1$.}
Hence $\CC(\Gamma)$ is always finite-dimensional; i.e., it is a ``$\dagger$-algebra'', in physics terminology.
We will frequently be interested in the case where $\p$ is self-adjoint, meaning $\p^* = \p$.
In this case, note that each $a_S$ will either be real (if $\X^S$ is equal to its reverse) or purely imaginary (if $\X^S$ equals the negation of its reverse).
Finally, we introduce the following notation (see \Cref{rem:trace} for justification):
\begin{notation}    \label{not:trace}
    For $\p \in \CC(\Gamma)$ as in \Cref{eqn:h} we write $\tr(\p) \coloneqq a_\emptyset$.
\end{notation}

\subsection{The completely anticommuting case} \label{sec:KN}
Let us collect here some simple facts concerning the case $\Gamma = K_{\N}$, where the indeterminates pairwise anticommute.
\begin{fact}                                        \label{fact:rev}
    In $\CC(K_{\N})$, monomial $(\X^S)^* = \sigma \X^S$, where
    \[
        \sigma = (-1)^{\binom{|S|}{2}} = \begin{cases} +1 & \text{if $|S| = 0,1$ mod $4$,} \\
                                                                                    -1& \text{if $|S| = 2,3$ mod $4$.} \end{cases}
    \]
    Consequently, we can write a generic self-adjoint homogeneous degree-$q$ polynomial as
    \begin{equation}    \label{eqn:typical}
        \p = \i^{\binom{q}{2}} \sum_{\substack{S \subseteq [\N] \\ |S| = q}} a_S \X^S,
    \end{equation}
    where each $a_S$ is real.
\end{fact}
\begin{fact}                                     \label{fact:ac}
    Let $\X^S$ and $\X^T$ be monomials in $\CC(K_{\N})$.
    Then
    \[
        \X^S \X^T = (-1)^{|S|\cdot|T| - |S \cap T|} \X^T \X^S.
    \]
    In particular, if $|S| = |T| = q$, then $\X^S$ and $\X^T$ anticommute iff $q - |S \cap T|$ is odd.
\end{fact}
\begin{fact}                                        \label{fact:ortho}
    Let $\ell_a = a_1 \X_1 + \cdots + a_\N \X_\N,\ \ell_b = b_1 \X_1 + \cdots + b_\N \X_\N \in \CC(K_\N)$ for $\vec a,\vec b \in \R^\N$.
    Then $\ell_a \ell_a = |\vec a|^2$, and $\ell_a \ell_b = -\ell_b \ell_a$.
    In particular, if $O \in \R^{\N \times \N}$ is an orthogonal matrix and $\ell = O\X$, then $\ell_1, \dots, \ell_\N$ are  self-adjoint, square to~$1$, and pairwise anticommute.
    Indeed, they generate an isomorphic copy of~$\CC(K_\N)$.
\end{fact}

\subsection{Representations}        \label{sec:represents}
By a \emph{representation} of $\CC(\Gamma)$, we will mean a $*$-homomorphism $\pi : \CC(\Gamma) \to \C^{D \times D}$ (for some finite~$D$) satisfying $\pi(1) = \Id$.
Informally, $\pi$ is an assignment of matrices (of some fixed dimension~$D$) to each~$\X_j$ such that the assigned matrices satisfy the algebra relations.
The representation theory of quasi-Clifford algebras is completely understood; Gastineau-Hills~\cite{GH82} showed that $\CC(\Gamma)$ is always a semisimple algebra and determined its irreducible representations (see also the concrete analysis in~\cite{Kho08}).
We recap it here.

First suppose $\N = 2r$ is even and $\Gamma = ([\N], \{\{1,2\}, \{3,4\}, \dots, \{\N-1,\N\}\})$, the matching graph on~$\N$ vertices.
Then we have an (irreducible) $2^r$-dimensional representation~$\pi$ that assigns
\newcommand{\II}{\makebox[\widthof{$X$}][c]{$\Id$}}
\newcommand{\YY}{\makebox[\widthof{$X$}][c]{$Y$}}
\begin{align*}
    \X_1 &= X \otimes \II \otimes \II \otimes \cdots \otimes \II \\
    \X_2 &= \YY \otimes \II \otimes \II \otimes \cdots \otimes \II \\
    \X_3 &= \II \otimes X \otimes \II \otimes \cdots \otimes \II \\
    \X_4 &= \II \otimes \YY \otimes \II \otimes \cdots \otimes \II \\
    &\,\cdots\\
    \X_{\N-1} &= \II \otimes \II \otimes \II \otimes \cdots \otimes X \\
    \X_{\N} &= \II \otimes \II \otimes \II \otimes \cdots \otimes \YY, \\
\end{align*}
where we use the notation $X = \begin{pmatrix} 0 & 1 \\ 1 & 0 \end{pmatrix}$, $Y = \begin{pmatrix} 0 & -\i \\ \i & \phantom{-}0 \end{pmatrix}$, $Z = \begin{pmatrix} 1 & \phantom{-}0 \\ 0 & -1 \end{pmatrix}$ for the Pauli matrices.
In fact this~$\pi$ yields an isomorphism between $\CC(\Gamma)$ and $\C^{2^r \times 2^r}$.

As a next example, consider $\CC(\ol{K}_\N)$, the algebra in which all indeterminates~$\X_j$ pairwise commute.
Then for each $x \in \{\pm 1\}^{\N}$ we have a $1$-dimensional representation $\pi_x$ in which $\X_j$ is assigned~$x_j$.
The direct sum $\pi = \bigoplus_x \pi_x$ is a (reducible) $2^{\N}$-dimensional representation, which can alternately be expressed as the assignment
\begin{equation}    \label{eqn:x}
    \X_1 = Z \otimes \Id \otimes \cdots \otimes \Id, \quad
    \X_2 = \Id \otimes Z \otimes \cdots \otimes \Id, \quad \dots, \quad
    \X_\N = \Id \otimes \Id \otimes \cdots \otimes Z,
\end{equation}
and this $\pi$ yields a $*$-homomorphism 
$\CC(\Gamma) \to \C^{D \times D}$.

We can also combine the previous two constructions.
If $\Gamma$ consists of $r$ disjoint edges and $s$~isolated vertices (so $\N = 2r+s$), then there is a $*$-homomorphism to $\C^{D \times D}$ for $D = 2^{r + s} = 2^{\N - r}$, via the representation~$\pi$ which maps the indeterminates to matrices of the form
\begin{gather*}
    (\Id \otimes \cdots \Id \otimes X \otimes \Id \otimes \cdots \otimes \Id) \otimes (\Id \otimes \cdots \otimes \Id) \text{ or }
    (\Id \otimes \cdots \Id \otimes Y \otimes \Id \otimes \cdots \otimes \Id) \otimes (\Id \otimes \cdots \otimes \Id) \\
    \text {or }
    (\Id \otimes \cdots \otimes \Id) \otimes (\Id \otimes \cdots \Id \otimes Z \otimes \Id \otimes \cdots \otimes \Id).
\end{gather*}

As a next observation, suppose $\Gamma$ has adjacency matrix~$A$, regarded as an element of $\F_2^{\N \times \N}$.
Then if we replace the $j$th row with the $j$th plus the $k$th, and do the same columnwise, the resulting $A'$ is the adjacency matrix of a graph~$\Gamma'$ with $\CC(\Gamma') \cong \CC(\Gamma)$, via the isomorphism that has $\X'_j = \i^{A_{jk}} \X_j \X_k$ and $\X'_{j'} = \X_{j'}$ for $j' \neq j$.
By repeating such operations, any $\Gamma$ can be transformed to another $\Gamma'$ consisting of $r$ disjoint edges and $s$ isolated vertices, where $2r$ is the $\F_2$-rank of the original~$A$.
In this way we can obtain an explicit $*$-homomorphism to a matrix algebra for any particular~$\CC(\Gamma)$.

Let us turn to the case of most interest to us, the usual Clifford algebra~$\CC(K_{\N}$), where we assume for convenience that $\N = 2r$ is even.
In this case, we obtain the well known $2^{r}$-dimensional representation~$\pi$ of~$\CC(K_{\N})$ given by the Weyl--Brauer $\gamma$~matrices:
\begin{equation}    \label{eqn:gamma}
    \gamma_{2k-1} = \underbrace{Z \otimes \cdots \otimes Z}_{k-1 \text{ times}} \otimes X \otimes
    \underbrace{\Id \otimes \cdots \otimes \Id}_{r - k \text{ times}}, \qquad
    \gamma_{2k} = \underbrace{Z \otimes \cdots \otimes Z}_{k-1 \text{ times}} \otimes Y \otimes \underbrace{\Id \otimes \cdots \otimes \Id}_{r - k \text{ times}}, \qquad k = 1 \dots r = \N/2.
\end{equation}

Summarizing these results (and applying standard theory, e.g.~\cite[Cor.~III.1.2]{Dav96}):
\begin{proposition}                                     \label{prop:reps}
    (\cite{GH82}.)
    Given graph $\Gamma$ on~$\N$ vertices, let $2r$ be the rank of its adjacency matrix over~$\F_2$, let $s = \N-2r$, and let $D = 2^{\N-r}$.
    Then we may define a canonical $D$-dimensional representation
$\pi : \CC(\Gamma) \to \C^{D \times D}$.
    This $\pi_{\Gamma}$ may be expressed as $\bigoplus_{x \in \{\pm 1\}^s} \pi_{\Gamma,x}$, where each~$\pi_{\Gamma,x}$ is irreducible and isomorphic to the Weyl--Brauer representation for~$\CC(K_{2r})$ from \Cref{eqn:gamma} (and hence is of dimension $2^{r}$).
    Finally, \emph{any} representation of $\CC(\Gamma)$ is unitarily equivalent to a direct sum $\bigoplus_{j=1}^m \pi_j$, where each~$\pi_j$ is a copy of one of the $\pi_{\Gamma,x}$'s.
\end{proposition}

\begin{remark} \label{rem:trace}
    It follows from \Cref{prop:reps} that for any $\p \in \CC(\Gamma)$ and any representation~$\pi$ we have $\tr(\pi(\p)) = \tr(\p)$ (as in \Cref{not:trace}).
    In particular, $\tr(gh) = \tr(hg)$ always.
\end{remark}

\subsection{States, pseudostates, and pseudoexpectations}
We recall some standard terminology:
\begin{definition}
    We say that $\p \in \CC(\Gamma)$ is \emph{positive} if and only if it is a (Hermitian-)square, $\p = g^* g$ for some $g \in \CC(\Gamma)$.
\end{definition}
\begin{definition}  \label{def:abstract-state}
    Let $\varrho : \CC(\Gamma) \to \C$ be a linear functional (which we take to include the condition $\varrho(\p^*) = \varrho(\p)^*$, so that $\varrho(\p) \in \R$ when $\p$ is self-adjoint).
    We say that $\varrho$ is \emph{positive} if $\varrho(\p) \in \R^{\geq 0}$ for all positive~$\p$.
    We furthermore say that $\varrho$~is a \emph{state} if $\varrho(1) = 1$.
\end{definition}
Under the canonical ~$\pi_\Gamma$ from \Cref{prop:reps}, we see that any linear functional~$\varrho$ is of the form $\p \mapsto \tr(\rho \pi_\Gamma(\p))$ for some self-adjoint $\rho \in \C^{D \times D}$.
From this we infer that $\varrho$ is positive if and only if~$\rho$ is positive semidefinite, and furthermore $\varrho$ is a state if and only if~$\tr(\rho) = 1$ --- i.e., if and only if $\rho$ is a quantum state.\footnote{In this work we use the nonstandard convention in which quantum states~$\rho$ have $\tr(\rho) = 1$; this is as opposed to the conventional normalization requiring $\Tr(\rho)=1$.}
Using $\pi_\Gamma$ to pull back to $\CC(\Gamma)$, we are led to an equivalent definition:
\begin{definition}
    Abusing terminology, we also say that $\rho \in \CC(\Gamma)$ is a \emph{state} if it is positive and has $\tr(\rho) = 1$.
    In this case we use the notation $\E_\rho[\p] = \tr(\rho \p)$, calling this the \emph{expectation of $\p$ under~$\rho$}.
\end{definition}
\begin{examples} \label[eg]{eg:pmstate}
    In the case of $\Gamma = \CC(\ol{K}_\N)$, where all indeterminates commute,
    for each $x \in \{\pm 1\}^\N$ the following is a state:
    \[
        \rho_x =  (1+x_1 \X_1)(1+x_2\X_2) \cdots (1+x_\N \X_N).
    \]
    It satisfies $\E_{\rho_x}[\p] = \p(x)$.
    Under the canonical representation we have that $\pi_\Gamma(\rho_x) \in \C^{2^\N \times 2^\N}$ has a single nonzero entry in the $(x,x)$ position, equal to $2^\N$.
\end{examples}

For the sake of optimization, we will be interested in relaxing the notion of a linear functional~$\varrho$ being positive.
We begin with the following definition:
\begin{definition}
    We say that $\p \in \CC(\Gamma)$ is a \emph{degree-$2k$ sum of (Hermitian-)squares} (abbreviated \emph{SOS}) if it is expressible as $g_1^* g_1 + \cdots + g_m^* g_m$ for some polynomials $g_1, \dots, g_m \in \CC(\Gamma)$ of degree at most~$k$.
\end{definition}
\begin{remark}
    Any positive $\p \in \CC(\Gamma)$ is evidently degree-$2k$ SOS for some~$k \leq n$ (see \Cref{eqn:h}).
    Conversely, any SOS~$\p$ is positive.
    This can be shown using the representation $\pi_\Gamma$ from \Cref{prop:reps}: if $\p$ is SOS then $\pi_\Gamma(\p)$ is a sum of (Hermitian)-squares of matrices, hence $\pi_\Gamma(\p)$ is PSD and thus of the form $G^* G$, hence $g \coloneqq \pi_{\Gamma}^{-1}(G)$ has $g^* g = \p$.
    Note, however, that $\p$ being low-degree SOS does not necessarily imply that $\p$ is the square of a low-degree polynomial.
\end{remark}
The following simple fact shows that being degree-$2k$ SOS is characterized by a semidefinite program of size~$\N^{O(k)}$ (or $2^{O(\N)}$ for $k = \Theta(\N)$):
\begin{fact} \label{fact:sdpk}
    Let $\vec X$ denote the vector whose entries are the monomials $\X^S$ for all $|S| \leq k$.
    Then $\p \in \CC(\Gamma)$ is degree-$2k$ SOS if and only if there is a complex PSD matrix~$H$ such that $\braket{X|H|X} = \p$ within $\CC(\Gamma)$.
\end{fact}

\begin{definition}
    We define a \emph{degree-$2k$ pseudoexpectation} to be a linear functional $\pE[\cdot] : \CC(\Gamma) \to \C$ with $\pE[1] = 1$ and $\pE[\p] \geq 0$ for all degree-$2k$ SOS~$\p$ (or equivalently, $\pE[g^* g] \geq 0$ for all $g$ of degree at most~$k$).
    As before, we may always write $\pE[\p] = \tr(\rho \p)$ for some self-adjoint~$\rho \in \CC(\Gamma)$.
    We therefore call a self-adjoint $\rho \in \CC(\Gamma)$ a \emph{degree-$2k$ pseudostate} if has $\tr(\rho) = 1$ and $\tr(\rho \p) \geq 0$ for all degree-$2k$ SOS~$\p$ (equivalently, $\tr(g \rho g^*) \geq 0$ for all $g$ of degree at most~$k$).
    In this case we use the notation $\pE_\rho[\p] = \tr(\rho \p)$.
\end{definition}
\begin{remark}
    Whether or not a $\pE[\cdot]$ is a degree-$2k$ pseudoexpectation is completely determined by the values $\pE[\X^S]$ for $|S| \leq 2k$.
\end{remark}
In light of the above, we have the following observation which shows that degree-$2k$ pseudoexpectations are characterized by a semidefinite program of size~$\N^{O(k)}$ (or $2^{O(\N)}$ for $k = \Theta(\N)$):
\begin{fact}                                        \label{fact:sossdp}
    Let $\pE[\cdot]$ be a linear functional on $\CC(\Gamma)$ with $\pE[1] = 1$, and let $M$ be the matrix having rows and columns indexed by the sets with $|S| \leq k$, and having $(S,T)$-entry equal to $\pE[(\X^S)^* \X^T]$.
    Then $\pE[\cdot]$ is a degree-$2k$ pseudoexpectation if and only if~$M$ is PSD.
\end{fact}

\subsection{Optimization}
Given a self-adjoint polynomial $\p \in \CC(\Gamma)$, we consider the following optimization problem:
\[
    \Opt(\p) \coloneqq \max_{\text{states $\rho \in \CC(\Gamma)$}}\{{\E}_\rho[\p]\}.
\]
Recalling $\pi_\Gamma$ from \Cref{prop:reps} and the discussion in the previous section, we equivalently have
\begin{align*}
    \Opt(\p) &= \max\{\tr(\rho \pi_\Gamma(\p)) : \rho \in \C^{D \times D},\ \rho \text{ PSD},\ \tr(\rho) = 1\}  \\
                 &= \max\{\braket{\phi|\pi_\Gamma(\p)|\phi} : \ket{\phi} \in \C^D,\ \braket{\phi|\phi} = 1\} \\
                 &= \lambda_{\max}(\pi_\Gamma(\p)).
\end{align*}
In physics terminology, if $\p$ is a Hamiltonian, then $\Opt(\p)$ is the ground state energy of~$-\p$ (under~$\pi_\Gamma$).
As a further equivalence, from \Cref{prop:reps} it follows that
\[
    \Opt(\p) = \max\{ \lambda_{\max}(\pi(\p)) : \text{representations } \pi \text{ of } \CC(\Gamma)\};
\]
in other words, the optimization problem seeks the assignment to the indeterminates~$\X_j$ (satisfying the \mbox{(anti-)}commutation relations~$\Gamma$) that produces a matrix with largest possible eigenvalue.
We will also sometimes consider the closely related problem of maximizing the operator norm; i.e., finding $\Opt_{\pm}(\p) = \max\{\Opt(\p), -\Opt(-\p)\}$.
In fact, this makes sense even for non-self-adjoint~$\p$, if we define
\[
    \Opt_{\pm}(\p) \coloneqq \sqrt{\Opt(\p^* \p)}.
\]

\begin{examples}    \label[eg]{eg:fix}
    In the case of most interest to us, $\CC(K_{\N})$ with $\N$~even, the representation~$\pi_\Gamma$ is the one using $\gamma$-matrices from \Cref{eqn:gamma}.
    This means that given the (coefficients of the) polynomial~$\p(\X)$,
    the task is to find the largest eigenvalue of the implicitly determined $2^{\N/2} \times 2^{\N/2}$ matrix~$\p(\gamma)$.
\end{examples}

\begin{examples} \label[eg]{eg:pm}
    In the case of $\CC(\ol{K}_\N)$, where all indeterminates commute, the task reduces to optimizing over all the $1$-dimensional representations~$\pi_x$ (described prior to \Cref{eqn:x}).
    That is, we are given a commutative polynomial $\p(\X)$ with real coefficients, and we are optimizing it over the Boolean hypercube,
    \[
        \Opt(\p) = \max_{x \in \{\pm 1\}^{\N}} \{\p(x)\}.
    \]
\end{examples}

\begin{remark}  \label{rem:naive}
    Supposing $\p = \sum_{|S| = q} a_S \X^S \in \CC(\Gamma)$ with $\vec a$ real, we may ask for a  ``naive'' upper bound on $\Opt_{\pm}(\p)$.
    In the fully commuting case, $\Gamma = \ol{K}_\N$, a reasonable answer is $|\vec a|_1$; this quantity is sometimes achievable (e.g., if each $a_S$ is nonnegative, or if $q = 1$), and in the computer science literature $|\vec a|_1 = 1$ is a typical normalization when $q = 2$ (``Max-Cut'').

    On the other hand, in the fully anticommuting case $\Gamma = K_\N$, it seems that normalizing by $|\vec a|\equiv |\vec a|_2 = 1$ is most natural.
    For example, when $q = 1$ we have $\p^*\p = |\vec a|^2$ (see \Cref{fact:ortho}) and hence $\Opt_{\pm}(\p) = |\vec a|$.
    In addition, normalization by $|\vec a|$ arises naturally in our generic upper bound using the \lovasz Theta Function, presented in \Cref{sec:lovaszth} (see \Cref{prop:lovasz-upper}).
\end{remark}

As we see from \Cref{eg:fix,eg:pm}, given a degree-$q$ polynomial $\p$ by its $\N^{O(q)}$ coefficients, the naive algorithms for determining $\Opt(\p)$ have exponential complexity.
This motivates understanding variants of~$\Opt(\p)$ that are computable in polynomial time, foremost of which is the following ``SOS relaxation'', which \Cref{fact:sossdp} shows is expressible as an SDP:
\begin{definition}
    Given a self-adjoint $\p \in \CC(\Gamma)$ of degree~$q$, as well as $k \in \NN^+$ with $2k \geq q$, we define
    \begin{equation}    \label{eqn:sos-primal}
        \SOS_{2k}(\p) = \max \{{\pE}_\rho[h] : \pE[\cdot] \textnormal{ is a degree-$2k$ pseudoexpectation}\}.
    \end{equation}
    This is a nonincreasing function of~$k$, with $\SOS_{2n}(\p) = \Opt(\p)$.
\end{definition}
Alternatively, we may seek to certify $\Opt(\p) \leq \beta$ by expressing $\beta - \p$ as a degree-$2k$ sum of (Hermitian-)squares.
Finding the smallest $\beta$ for which this is possible is the dual SDP for~$\SOS_{2k}$, as is well known.
Furthermore, using the fact that we have $\X_j^2 = 1$ for all~$j$, it is not hard to show that there is no duality gap (see, e.g.,~\cite{JH16}, with the noncommutativity of our setting not affecting the proof).
Thus:
\begin{theorem}
    Given a self-adjoint $\p \in \CC(\Gamma)$ of degree~$q$, as well as $k \in \NN^+$ with $2k \geq q$, we have
    \begin{equation}    \label{eqn:sos-dual}
        \SOS_{2k}(\p) = \min \{\beta : \beta - \p \textnormal{ is degree-$2k$ SOS}\}.
    \end{equation}
\end{theorem}
\Cref{fact:sdpk} explains why this characterization is also an SDP.

\subsection{Complexity theory}  \label{sec:complexity}
Regarding the optimization problem $\Opt(\p)$, there are several computational tasks we will want to consider.
But first we must cover some minor technical details.

\paragraph{Complexity-theoretic niceties.}
The input is always considered to be the coefficients of~$\p$, explicitly given as rational (complex) numbers, as well as the graph~$\Gamma$ if necessary (i.e., if it isn't considered fixed to be~$K_{\N}$ or~$\ol{K}_\N$).
Most often we have a fixed constant degree~$q$ in mind for~$\p$, in which case the number of coefficients of~$\p$ is bounded by~$m = O(\N^q)$.
For the sake of scaling, we typically prefer to assume the normalization $\sum_S |a_S|^2 = 1$ (where $\p$ is written as in \Cref{eqn:h}).
With rational coefficients perhaps only $\sum_S |a_S|^2 \approx 1$ will be possible, but anyway we will never be interested in cases other than $1/\poly(\N) \leq \sum_S |a_S|^2 \leq \poly(\N)$.
In light of this (and the fact that~$\|\pi(X^S)\| \leq 1$ for all~$S$), perturbing any coefficient~$a_S$ by a sufficiently small $\pm 1/\poly(\N)$ will have only negligible effect on~$\Opt(\p)$, and we may therefore assume that all coefficients~$a_S$ are expressed using $O(\log \N)$ bits.
(In particular, when we consider Gaussian coefficients~$J_S$, we may assume they are rounded to $O(\log \N)$ bits of precision.)
Thus the total input size~$L$ will generally be $O(\N^q \log \N)$, plus $O(\N^2)$ if~$\Gamma$ needs to be specified.
Finally, when we speak of ``computing $\Opt(\p)$'', or ``solving a semidefinite program'', we mean don't mean literally exactly, but rather up to an additive precision of any arbitrarily small~$1/\poly(\N)$.
Indeed, the SDP for $\SOS_{2k}(\p)$ can be computed to additive precision $\exp(-\N^k)$ in $\N^{O(k)}$ time.\\

Given an instance~$\p$ (and graph~$\Gamma$), there are several possible computational tasks concerning~$\Opt(\p)$ (or $\Opt_\pm(\p)$):

\paragraph{Exact computation.} Determine~$\Opt(\p)$.  (But this is $\QMA$-hard.)

\paragraph{Upper-bound certification.} Output $\beta \in \R$ such that $\Opt(\p) \leq \beta$.

\paragraph{Lower-bound certification.} Output $\alpha \in \R$ such that $\Opt(\p) \geq \alpha$.\\

Naturally, one wants $\beta$ as small as possible and $\alpha$ as large as possible.
One might also want to upgrade a lower-bound certification algorithm so that it also provides a \emph{witness} to $\Opt(\p) \geq \alpha$.
We consider two possibilities:

\paragraph{Implicit lower-bound witnessing.} Output a classical description\footnote{Here it is only fair to allow for a fixed ``encoding scheme'' $\textsc{Enc} : \{\text{bit-strings}\} \to \CC(\Gamma)$ with the following properties: (i)~$\textsc{Enc}(x)$ is a valid state for any bit-string~$x$; (ii)~there is an algorithm that, given $h \in \CC(\Gamma)$ by its coefficients as well as some~$x$, efficiently evaluates $\E_{\textsc{Enc}(x)}[h]$.
In the fully commuting case of $\CC(\ol{K}_\N)$, the mapping $x \mapsto \rho$ from \Cref{eg:pmstate} is natural.
In the fully anticommuting case of~$\CC(K_\N)$, there is an analogous natural scheme for encoding the \emph{Gaussian states}, as described in \Cref{sec:gaussian-states}.
The requirement (ii) is necessary as otherwise one could simply declare that
the state is $\exp(\beta p)\ket{\phi}$ for large $\beta>0$ and for $\phi$ chosen to have nontrivial projection on the eigenspace of maximal eigenvalue; as $\beta\rightarrow\infty$, this state converges to that eigenspace, but we cannot compute expectation values efficiently.  Another interesting state that does not seem to meet our requirements is the coupled cluster variational state of \cite{HTS21}.
}
 of a state~$\rho$ achieving~\mbox{$\E_\rho[\p] \geq \alpha$}.

\paragraph{Quantum lower-bound witnessing.} Output an actual quantum state $\pi_\Gamma(\rho) \in \C^{D \times D}$ of $\N-r$ qubits (under the canonical representation from \Cref{prop:reps}) achieving $\E_\rho[\p] \geq \alpha$.\\

The reader should bear in mind that for the implicit lower-bound witnessing problem (with fixed encoding scheme), there may not even \emph{be} a state~$\rho$ that simultaneously has a $\poly(\N)$-size description and a particularly large value for $\E_\rho[\p]$.
Similarly, for the quantum lower-bound witnessing problem, for a given~$\p$ there may not \emph{be} a $\poly(\N)$-size quantum circuit that outputs a~$\rho$ with particularly large~$\E_\rho[\p]$.
This latter possibility is related to the fermionic NLTS Conjecture discussed in \Cref{sec:cc}.

\subsubsection{On randomized algorithms and randomized inputs}  \label{sec:random}
In this work we will often be interested in \emph{average-case} complexity, where the input for the task is~$\boldp$ drawn from some distribution~$\calH$ on instances (typically $\calH = \SYK_q$ for some~$q$).
In this case we should be careful to define the correctness of algorithms.

Except for one case (to be discussed shortly), each certification/witnessing algorithm presented in this paper is deterministic.
Such a deterministic algorithm ALG will have the following properties:
On all inputs~$\p$, algorithm ALG outputs either a \emph{claim} (such as ``$\Opt(\p) \leq \sqrt{\N}$'') or else outputs ``FAIL''.
And furthermore:
\begin{enumerate}[label=(\roman*)]
    \item For \emph{all} inputs $\p$, if ALG$(\p)$ is a claim, then that claim is correct.
    \item $\Pr_{\boldp \sim \calH}[\text{ALG}(\boldp) = \text{``FAIL''}] < \eps$ (for some small~$\eps$).
\end{enumerate}
We emphasize two things:
On one hand, the user of the ALG gets $100\%$ confidence in its claims.
On the other hand, for a given outcome $\boldp = \p$, a user \emph{cannot} ``boost'' the probability of a non-FAIL outcome by repeatedly running the algorithm.
(We also remark that for the task of implicit lower-bound witnessing, there is never a concern about confidence; we assume that from the output description of~$\rho$ one can efficiently and deterministically compute~$\E_{\rho}[\p]$.)\\

The one exception to this deterministic style of algorithm will be our quantum lower-bound witnessing for $\boldp \sim \SYK_4(\N)$.
In \Cref{sec:variational} we first show that
\[
    \Pr_{\boldp \sim \SYK_4(\N)}[\Opt(\boldp) \geq c_0\sqrt{\N}] \geq 1 - \exp(-\Omega(\N))
\]
(for some universal constant $c_0 > 0$).
Following that we give an efficient quantum ``witnessing'' algorithm~$W$ that, given any input~$\p$, applies a sequence of polynomially-many gates to produce a quantum state~$\rho$.
One may think of~$W$ \emph{hypothesizing} that $\E_{\rho}[\p] \geq c_0\sqrt{\N}$.

Now on one hand, there is no obvious \emph{deterministic} way to test this hypothesis, even with a quantum computer.
On the other hand, there \emph{is} an obvious \emph{randomized} certification algorithm~$\bV$ that runs efficiently on a quantum computer.
Given~$\p$, algorithm $\bV$ simply repeatedly prepares $W(\p) = \rho$ and empirically estimates~$\E_\rho[\p]$.
By running $\bV$ for $\poly(\N)\log(1/\delta)$ time, a user can obtain confidence $1-\delta$ that the hypothesis $\Opt(\p) \geq c_0\sqrt{\N}$ is true (up to any additive~$1/\poly(\N)$, which we can neglect since we aren't specifying~$c_0$ precisely anyway).
That is, for any user-selected~$\delta$, we can obtain a $\poly(\N)\log(1/\delta)$-time composite quantum algorithm \textbf{ALG} that, on input~$\p$, either outputs ``FAIL'', or else outputs a claim of $\Opt(\p) \geq c_0\sqrt{\N}$ together with a witnessing quantum state~$\rho$.
It has the following properties:
\begin{enumerate}[label=(\roman*)]
    \item[(i${}'$)] For \emph{all} inputs $\p$, if $\textbf{ALG}(\p)$ outputs a claim, then that claim is correct except with probability at most~$\delta$ (over the internal randomness of~$\textbf{ALG}$).
    \item[(ii${}'$)] $\Pr_{\boldp \sim \SYK_4(\N)}[\textbf{ALG}(\boldp) = \text{``FAIL''}] < \exp(-\Omega(\N))$.
\end{enumerate}
In summary, with very high probability over the random \emph{input}~$\boldp$, the algorithm produces a~$\rho$ and the claim that $\E_{\rho}[\boldp] \geq \Omega(\sqrt{\N})$.
Furthermore, for \emph{any} input~$\p$ resulting in a claim, the user can verify this claim with confidence that can be efficiently tuned arbitrarily close to~$1$.\\

We end this section by remarking that for $\boldp \sim \SYK_4(\N)$, we currently do not know any efficient \emph{classical} algorithm for certifying $\Opt(\boldp) \geq \Omega(\sqrt{\N})$ (even without a witness)  that has properties $(\text{i}'), (\text{ii}')$ above.
Thus we arguably have a distributional task, with classical input and output, that is efficiently solvable by a quantum computer, but for which we currently do not know an efficient classical algorithm.

\section{Upper bounds and heuristics for random polynomials}
In this section we recall (and generalize) the known upper bound on $\E[\Opt(\bh)]$ for $\bh \sim \SYK_q(\N)$, and also sketch the heuristic from mathematical physics that suggests its asymptotic value.

Before this, we  recall that analyzing the expected optimum is sufficient for understanding the high-probability optimum.
To be precise, note for $\bp \sim \SYK_q(\N)$ that $\Opt(\bp)$ and $\Opt_\pm(\bp)$ are $1/\sqrt{\binom{\N}{q}}$-Lipschitz functions of the random coefficients.
Thus standard concentration results for Gaussian measure (see, e.g.,~\cite[Ch.~2]{Led01}) imply the following:
\begin{proposition}                                     \label{prop:concentrate}
    For $\bp \sim \SYK_q(\N)$ it holds that
    \[
        \Pr[|\Opt(\bp) - \E[\Opt(\bp)]| \geq \eps] \leq 2\exp(-\tbinom{\N}{q} \eps^2/2)
    \]
    (and the same holds for $\Opt_\pm(\bp)$).
    E.g., for any constant $q \geq 2$ there is a constant~$C$ such that
    \[
        \Pr[|\Opt(\bp) - \E[\Opt(\bp)]| \geq C/\sqrt{\N}] \leq \exp(-\N).
    \]
\end{proposition}

\subsection{Upper bounds for random polynomials}
On the subject of upper bounds for $\Opt(\bp)$ when $\bp \sim \SYK_q(\N)$, the best provable upper bound that is known is $\E[\Opt_{\pm}(\bp)] \leq \sqrt{\ln 2} \sqrt{\N}$.
One way to prove this\footnote{As was pointed out to us by several attendees, including Jerry Li and Cristopher Moore, at the Simons Institute workshop \emph{Rigorous Evidence for Information-Computation Trade-offs}.} is to use the Matrix Chernoff Bound~\cite{AW02,Oli10}.
Since the $\gamma$-matrices from~\Cref{eqn:gamma} are in dimension $D = 2^{\N/2}$, and since we have $\|\gamma_{j_1} \cdots \gamma_{j_q}\|_{\mathrm{op}} \leq 1$ always, the bound~\cite[Thm.~1.2]{Tro12} gives
\[
    \Pr_{\bp \sim \SYK_q(\N)}[\Opt(\bp) \geq t] \leq 2^{\N/2} \exp(-t^2/2).
\]
From this, $\E[\Opt(\bp)] \leq (\sqrt{\ln 2} + o(1))\sqrt{\N}$ follows (and we only need $\bp$ to have centered, subgaussian random coefficients, of total variance~$1$).
Indeed, Feng, Tian, and Wei~\cite{FTW19} obtain the exact bound $\E[\Opt_{\pm}(\bp)] \leq \sqrt{\ln 2}\sqrt{\N}$ via a simple proof, which is possible because the matrices $\gamma_{j_1} \cdots \gamma_{j_q}$ have uniformly bounded operator norm.
In the following \Cref{thm:generic-upper}, we give a generalization of this proof, which we will need (ironically) for our \emph{lower} bound in \Cref{sec:variational}.
\begin{theorem}                                     \label{thm:generic-upper}
    Given graph~$\Gamma = ([\N], E)$, suppose that
    \[
        \boldp(\X) = \sum_S \bJ_S \X^S \in \CC(\Gamma)
    \]
    is self-adjoint,
    where the coefficients $(\bJ_S)_{S \subseteq [\N]}$ are real random variables (not necessarily independent).
    Then provided conditions (i),(ii) below hold, we have
    \[
        \Pr[|\boldp(1)| \geq \beta] \leq 1/D \quad \implies \quad \E[\Opt_{\pm}(\boldp)] \leq \beta,
    \]
    where $D$ is the dimension of the canonical representation $\pi_\Gamma$ from \Cref{prop:reps}.

    The conditions are:
    \begin{align*}
        \text{(i)} &~\Pr[|\boldp(1)| \geq \beta] \leq 1/D \text{ holds due to the moment bound } \E[\boldp(1)^k]/\beta^k \leq 1/D \text{ for some even $k$;} \\
        \text{(ii)} &~\E[\bJ_{S_1} \bJ_{S_2} \cdots \bJ_{S_k}] \geq 0\  \forall S_{1}, \dots, S_k \subseteq [\N].
    \end{align*}
\end{theorem}
\begin{proof}
    This follows from
    \begin{multline*}
        \E[\Opt_{\pm}(\boldp)]^k \leq \E[\Opt_{\pm}(\boldp)^k] = \E[\|\pi_\Gamma(\boldp)\|^k] \leq \E[\Tr(\pi_\Gamma(\boldp)^k)] = D \E[\tr(\pi_\Gamma(\boldp^k))] \\
        = D\sum_{S_1, \dots, S_k \subseteq [\N]} \E[\bJ_{S_1} \cdots \bJ_{S_k}] \tr(\pi_\Gamma(\X^{S_1} \cdots \X^{S_k})) \overset{\text{(ii)}}{\leq} D \sum_{S_1, \dots, S_k \subseteq [\N]} \E[\bJ_{S_1} \cdots \bJ_{S_k}] = D \E[\boldp(1)^k] \overset{\text{(i)}}{\leq} \beta^k,
    \end{multline*}
    where the inequality marked ``(ii)'' also uses that each normalized-trace is at most~$1$, since $\pi_\Gamma(\X^{S_1} \cdots \X^{S_k})$ is a product of matrices $\pi_\Gamma(\X_j)$ that square to the identity and hence have operator norm at most~$1$.
\end{proof}
\begin{remark}  \label{rem:better-g}
    In fact, the proof gives $\E[\Opt_{\pm}(\bp)^k] \leq \beta^k$, and hence $\Pr[\Opt_{\pm}(\bp) \geq \beta'] \leq (\beta/\beta')^k$ for any $\beta' \geq \beta$.
\end{remark}
\begin{corollary}                                       \label{cor:syk-upper}
    For $\boldp \sim \SYK_q(\N)$ we have $\E[\Opt(\boldp)] \leq (\sqrt{\ln 2} + o(1))\sqrt{\N}$.
    More generally, this holds provided the coefficients $(\bJ_S : |S| = q)$ of $\boldp$ are  independent, centered, and subgaussian with parameter~$1/\binom{\N}{q}$ (e.g., they may be independent Rademachers, times $1/\sqrt{\binom{\N}{q}}$).
\end{corollary}
\begin{proof}
    In the setting of \Cref{thm:generic-upper} we have $\Gamma = K_\N$ and so $D = 2^{\N/2}$.
    Condition~(ii) holds by the independence and centeredness ($\E[\bJ_S] = 0$) of the coefficients.
    In the usual Gaussian case, $\boldp(1)$ is just a standard Gaussian random variable, hence it has $\Pr[|\boldp(1)| \geq \beta] \leq 2^{-\N/2}$ for $\beta \sim\sqrt{\ln 2} \sqrt{\N}$.
    Further, this can be shown by taking $k \sim (\ln 2)\N$ to be an even integer and using the formula $\E[\boldp(1)^k] = (k-1)!! \leq \sqrt{2} (k/e)^{k/2}$.
    The subgaussian case also follows.
\end{proof}
\begin{remark}
    The bound $O(\sqrt{\N})$ in \Cref{cor:syk-upper} is optimal up to constant factors for $q = 2,4$ (we prove this for $q=4$ in \Cref{sec:variational}.
    It is also expected to be optimal up to constant factors for larger even~$q$.
    However, it is not expected to be optimal for odd~$q$, where it is expected that $\E[\Opt(\boldp)] \leq O(1)$.
    See the heuristics discussed in \Cref{sec:heur}.
\end{remark}
\begin{remark}
    If the coefficients $\bJ_S$ are not subgaussian but have reasonably small moments, similar bounds can be obtained.
\end{remark}
\begin{remark}  \label{rem:FTW}
    For $\SYK_4(\N)$, Feng--Tian--Wei~\cite{FTW19} obtained the exact bound $\E[\Opt(\boldp)] \leq \sqrt{\ln 2} \sqrt{\N}$ by performing the proof with a convex combination of all integer~$k$.
\end{remark}
\begin{remark}
    As with upper bounds for random classical CSPs based on Chernoff, this upper bound is an excellent example of the sort that does \emph{not} obviously yield SOS proofs, or indeed any kind of succinct certificate.
    It shows that before we draw $\boldp$, we know it is very likely --- say, a $99$\% chance --- that $\Opt_{\pm}(\boldp)$ will be at most~$100\beta$.
    However once we \emph{fix} a randomly drawn outcome~$\boldp = \p$, there is no obvious certificate that this \emph{particular}~$\p$ has $\Opt_{\pm}(\p) \leq 100 \beta$.
\end{remark}
\begin{remark}  \label{rem:seq}
    In \Cref{thm:generic-upper} we assumed that the coefficients of $\boldp$ are indexed by \emph{subsets} $S \subseteq [\N]$, with $\X^S$ standing for $\X_{k_1} \cdots \X_{k_q}$ when $S = \{k_1, \dots, k_q\}$ with $k_1 < \cdots < k_q$.
    However one can see that the proof goes through equally well if~$S$ is allowed to stand for an \emph{ordered sequence} from~$[\N]$.
    Allowing this may assist in fulfilling condition~(ii).
\end{remark}
One scenario where we will be using \Cref{thm:generic-upper} (together with \Cref{rem:seq}) is the following:
$\boldp$ is a certain random polynomial of low degree, where each coefficient is of the form $c \bG_{j_1} \cdots \bG_{j_r}$, where $c \geq 0$ and $\bG_1, \dots, \bG_N$ are independent standard Gaussians.
In this case, condition~(ii) of \Cref{thm:generic-upper} holds, and we are reduced to understanding tail/moment bounds for polynomials in Gaussians (of the usual, commutative) kind.
For this, we will use Lata{\l}a's Theorem on Gaussian chaoses~\cite{Lat06} (see also~\cite{Leh11}):
\begin{theorem}                                     \label{thm:latala}
    (\cite{Lat06}.)
    Let $A \in \R^{N \times \cdots \times N}$ be a $q$-dimensional array of reals, with $A_{j_1 \cdots j_q} = 0$ when $j_1, \dots, j_q$ are not all distinct; and, let
    \begin{equation}    \label{eqn:dec}
        f(x_1, \dots, x_N) = \sum_{j_1, \dots, j_q = 1}^N A_{j_1 \dots j_q} \prod_{t=1}^q x_{j_t}.
    \end{equation}
    If $\bG = (\bG_1, \dots, \bG_N)$ denotes a sequence of independent standard Gaussians, then
     for any real $k \geq 2$,
    \[
        \|f(\bG)\|_k = \Ex[|f(\bG)|^k]^{1/k} \leq c_q \cdot \max_{\calP}\{k^{|\calP|/2} \|A\|_{\calP}\},
    \]
    where $c_q$ is a constant depending only on~$q$, the maximum is over all partitions~$\calP = (P_1, \dots, P_s)$ of~$[q]$ into nonempty parts, and
    \[
        \|A\|_{(P_1, \dots, P_s)} \coloneqq \max\braces*{f(x) : x \in \R^N,\ \sum_{p \in P_r} x_p^2 \leq 1\ \forall r \in [s]}.
    \]
    Indeed we have $\|f(\bG)\|_k \leq c'_q \|\tilde{f}(\bG^{(1)}, \dots, \bG^{(q)})\| \leq c_q \max_{\calP}\{k^{|\calP|/2} \|A\|_{\calP}\}$, where $\tilde{f}$ denotes the \emph{decoupled} version of \Cref{eqn:dec}, in which~$x_{j_t}$ is replaced with the new variable~$x_{j_t}^{(t)}$ (and $\bG^{(1)}, \dots, \bG^{(q)}$ are independent sequences of $N$~independent standard Gaussians).
\end{theorem}

\subsection{Heuristics for the SYK optimum} \label{sec:heur}
Here we describe the heuristic developed in the physics community for predicting $\Opt(\bh)$ when $\bh \sim \SYK_q(\N)$.
We follow the discussion in~\cite{GJV18}; see also~\cite{FTW19,FTW18,FTW20} for partial formalizations of this heuristic.

Naturally, one can attempt to analyze $\Opt_\pm(\bh)$, and indeed the entire typical spectrum of~$\bh$, via the (expected) Trace Method.
However, the combinatorics arising when computing $\E[\tr(\bh^k)]$ are rather intricate.
A heuristic one can use is the following.
From \Cref{fact:ac} we see that if $\bS$ and $\bT$ are chosen randomly and independently from the $N \coloneqq \binom{\N}{q}$ possible monomials, they anticommute with probability
\[
    p = p(q,\N) = \Pr[q - |\bS \cap \bT| \text{ is odd}] \approx
    \begin{cases}
        \Pr[|\bS \cap \bT| = 1] \approx q^2/\N & \text{for $q = o(\sqrt{\N})$ even,}\\
        \Pr[|\bS \cap \bT| = 0] \approx 1 - q^2/\N & \text{for $q = o(\sqrt{\N})$ odd,}\\
    \end{cases}.
\]
(Of course, one can easily make this asymptotic approximation more precise if desired.)
Given this, the key ansatz is to let $\bGamma \sim \mathfrak{G}(N,p)$ denote an Erd{\H o}s--R{\' e}nyi random graph, and to model the $N$~monomials in $\CC(K_{\N})$ by the $N$~indeterminates of~$\CC(\bGamma)$.
That is, for
\[
    \bell = \frac{1}{\sqrt{N}} \sum_{j=1}^N \ba_j \chi_j \in \CC(\bGamma), \qquad \ba_j \sim \mathrm{N}(0,1),
\]
the heuristic is
\begin{equation}    \label{eqn:tract}
    \E[\tr(\bh^k)] \approx \E_{\bGamma, \ba_i\text{'s}}[\tr(\bell^k)].
\end{equation}
Estimating the right side of \Cref{eqn:tract} is more computationally tractable for constant~$k$, and the resulting moment estimates match those of the so-called \emph{$\mu$-Gaussian} (or \emph{$\mu$-semicircular}) distribution, with orthogonal polynomials the \emph{$\mu$-Hermite polynomials} (see, e.g.,~\cite{ISV87}).%
\footnote{These are usually called \emph{$q$-Hermite polynomials} in the combinatorics literature, according to the usual notion of $q$-deformity.
Unfortunately we have a notational clash with the SYK degree~``$q$''.
In the physics literature this is resolved by writing ``$Q$-Hermite'', but we follow the ``$\mu$-Hermite'' terminology of Speicher~\cite{Spe92} instead.}
Here
\[
    \mu = \mu(q,n) \coloneqq 1-2p = \E_{\bS, \bT \sim \binom{[n]}{q}}[(-1)^{q - |\bS \cap \bT|}] \in [-1,+1],
\]
and one easily check that for $q = O(\sqrt{n})$, up to lower-order terms we have
\[
    \mu \approx (-1)^q \exp(-2r), \qquad r \coloneqq q^2/n. 
\]
Thus the heuristic prediction is that the spectral density of~$\bh \sim \SYK_q(n)$ will resemble that of the $\mu$-Gaussian density, call it $\nu_\mu$.  
This density $\nu_\mu$ has an explicit analytic formula (see, e.g.,~\cite{BKS97}), but rather than reproducing it, we will merely note the following facts:
\begin{itemize}
    \item For $\mu \to 1$ (corresponding to $q = o(\sqrt{n})$ even), the density $\nu_\mu$ tends to the standard Gaussian.
    \item For $\mu \in (-1,1)$, the support of $\nu_\mu$ is the interval $[-\frac{2}{\sqrt{1-\mu}}, +\frac{2}{\sqrt{1-\mu}}]$.
    \item For $\mu = 0$ (equivalently, $q = \omega(\sqrt{n})$), $\nu_\mu$ is the semicircle density.
    \item For $\mu \to -1$ (corresponding to $q = o(\sqrt{n})$ odd), the density $\nu_\mu$ tends to Rademacher.
\end{itemize}
Summarizing these heuristics, the belief is that with high probability,  
\begin{equation}    \label{eqn:heur}
    \text{for $\bh \sim \SYK_q(n)$,} \quad
    \Opt_{\pm}(\bh) \sim \frac{2}{\sqrt{1-(-1)^q\exp(-2q^2/n)}} \sim
    \begin{cases}
        \frac{\sqrt{2}}{q}\sqrt{n} & \text{for $q = o(\sqrt{\N})$ even,} \\
        \sqrt{2} & \text{for $q = o(\sqrt{\N})$ odd.}
    \end{cases}
\end{equation}

\section{Upper bound certification with the \lovasz Theta Function}    \label{sec:lovaszth}
We first recall the independence number of a graph:
\begin{notation}
    If $\Gamma = ([\N], E)$ is a graph, $\alpha(\Gamma)$ denotes the cardinality of its largest independent set; i.e., the largest $S \subseteq [\N]$ such that no $\{j,k\} \in E$ has both $j,k \in S$.
\end{notation}
\noindent In the context of $\CC(\Gamma)$, an independent set~$S$ is one for which the indeterminates $(\X_j : j \in S)$ pairwise commute.
\lovasz~\cite{Lov79} introduced the following upper bound on $\alpha(\Gamma)$, which is efficiently computable using semidefinite programming:
\begin{notation}
    The \emph{\lovasz Theta Function} assigns to each graph $\Gamma = ([\N], E)$ the number
    \begin{equation}    \label{eqn:primal}
        \th(\Gamma) = \max_{\rho \in \R^{\N \times \N}} \{ \Tr(\rho \J) : \rho \geq 0,\ \Tr(\rho) = 1,\ \rho_{jk} = 0\ \forall \{j,k\} \in E\},
    \end{equation}
    where $\J$ denotes the all-$1$'s matrix.
\end{notation}
\begin{fact}                                        \label{fact:th-sdp}
    There is also an equivalent dual SDP definition,
    \begin{equation}    \label{eqn:dual}
        \th(G) = \min \{\kappa : \exists \Sigma \in \R^{\N \times \N},\ \Sigma_{jj} = 1\ \forall j \in [\N],\ \Sigma_{jk} = 0\ \forall \{j,k\} \not \in E,\ \kappa \Sigma \geq \J\}.
    \end{equation}
\end{fact}
\begin{remark}
    The ``weak duality'' bound $\eqref{eqn:primal} \leq \eqref{eqn:dual}$, which is all we will require, is evident from
    \[
        \Tr(\rho \J) \leq \kappa \Tr(\rho \Sigma) = \kappa \Tr(\rho) = \kappa.
    \]
\end{remark}
\begin{fact}                                        \label{fact:alpha-theta}
    $\alpha(\Gamma) \leq \th(\Gamma)$ always holds: if $S$ is an independent set, and $\ket{S} \in \{0,1\}^{\N}$ denotes its indicator, then $\rho = \frac{1}{|S|}\ketbra{S}{S}$ achieves~$|S|$ in \Cref{eqn:primal}.
\end{fact}

Given a graph $\Gamma$, we will be interested in the maximum that $\Opt_\pm(\ell)$ can possibly be among self-adjoint linear polynomials $\ell \in \CC(\Gamma)$ (subject to a scaling constraint on their coefficients).
We therefore introduce the following:
\begin{definition}
    For a graph $\Gamma = ([\N],E)$, we define
    \[
        \OurPsi(\Gamma) = \max\braces*{\Opt_\pm(\ell)^2 = \Opt(\ell^*\ell) = \Opt(\ell^2) : \ell = \sum_{j=1}^{\N} a_j \X_j \in \CC(\Gamma),\ \vec a \in \R^{\N},\ |\vec a|^2 = 1}.
    \]
\end{definition}

The below two propositions now establish the following sandwiching bound:
\[
        \alpha(\Gamma) \leq \OurPsi(\Gamma) \leq \th(\Gamma).
\]
\begin{proposition}                                     \label{prop:alpha-lower}
    For any graph $\Gamma = ([\N],E)$ we have $\OurPsi(\Gamma) \geq \alpha(\Gamma)$.
\end{proposition}
\begin{proof}
    If $S$ is a maximum independent set in~$\Gamma$, then the matrices $(\pi_\Gamma(\X_j) : j \in S)$ pairwise commute and hence can be simultaneously diagonalized.
    Take any common eigenvector; it will have associated eigenvalue $\lambda_j \in \{\pm 1\}$ for $j \in S$.
    Now $\ell = \frac{1}{\sqrt{|S|}} \sum_{j \in S} \lambda_j \X_j$ has an eigenvalue of~$\sqrt{|S|} = \sqrt{\alpha(\Gamma)}$  under~$\pi_\Gamma$.
\end{proof}
\begin{proposition}                                     \label{prop:lovasz-upper}
    For any graph $\Gamma = ([\N],E)$ we have $\OurPsi(\Gamma) \leq \th(\Gamma)$, and in fact this has a degree-$2$ SOS proof.
    That is, if $\ell = \sum_{j=1}^{\N} a_j \X_j \in \CC(\Gamma)$ with $\vec a \in \R^{\N}$ and $|\vec a| = 1$, then $\SOS_2(\ell^*\ell) \leq \th(\Gamma)$.
    Moreoever, this holds even without assuming the commutation relations $\X_j \X_k = \X_k \X_j$ for $\{j,k\} \not \in E$; i.e., it only uses the anticommutation relations $\X_j \X_k = -\X_k \X_j$ for $\{j,k\} \in E$.
\end{proposition}
\begin{proof}
    Let $\pE$ be any degree-$2$ pseudoexpectation for~$\CC(\Gamma)$, which we may also think of as a PSD matrix in $\C^{\N \times \N}$.
    Now given $\ell = \sum_{j=1}^{\N} a_j \X_j$, with $|\vec a| = 1$, define the PSD matrix $\rho = \ketbra{a}{a} \circ \pE$, where $\circ$ denotes the Hadamard (entry-wise) product; i.e.,
    \[
        \rho_{jk} = a_j a_k \pE[\X_j \X_k].
    \]
    Then $\Tr(\rho) = |\vec a|^2 = 1$ and $\Tr(\rho \J) = \pE[\ell^* \ell]$.
    Furthermore, observe that for $\{j,k\} \in E$ we have $\pE[\X_j \X_k]^* = \pE[\X_k^* \X_j^*] = -\pE[\X_j\X_k]$; thus $\pE[\X_j \X_k]$ and $\rho_{jk}$ are purely imaginary.
    If we now let $\hat{\rho}$ be the real part of~$\rho$, we still have all of the following properties ---
    \[
        \hat{\rho} \geq 0, \quad \Tr(\hat{\rho}) = 1, \quad \Tr(\hat{\rho}J) = \pE[\ell^* \ell]
    \]
    --- and now $\hat{\rho}$ also has $\hat{\rho}_{jk} = 0$ for $\{j,k\} \in E$.
    This establishes $\pE[\ell^* \ell] \leq \th(\Gamma)$, as required.
\end{proof}

\begin{remark}
    Regarding tightness of the bounds $\alpha(\Gamma) \leq \OurPsi(\Gamma) \leq \th(\Gamma)$, the $5$-cycle $\Gamma = C_5$ is a case where $\OurPsi(\Gamma) < \th(\Gamma)$ strictly.
    In this case, famously $\th(C_5) = \sqrt{5}$~\cite{Lov79}, but one can show $\OurPsi(C_5) = 2 = \alpha(C_5)$.
    Briefly, in $\CC(C_5)$ we may write any $\ell$ under $\pi_\Gamma$ as
    \[
    \ell = a_1 \cdot X \otimes \Id \otimes \Id +
    a_2 \cdot Y \otimes \Id \otimes \Id +
    a_3 \cdot X \otimes X \otimes \Id +
    a_4 \cdot \Id \otimes Y \otimes \Id +
    a_5 \cdot Y \otimes Z \otimes Z.
    \]
    Then one may compute that the eigenvalues of $\ell^2$ are ($4$ copies each of) $\sigma\pm 2\sqrt{\tau}$, where $\sigma \coloneqq \sum_{j=1}^5 a_j^2$ and $\tau \coloneqq \sum_{\{j,k\} \not \in E}a_j^2 a_k^2$, with $E$ is the edge-set of $C_5$.  Our normalization gives $\sigma = 1$, and subject to this it is elementary to show that $\tau \leq \frac14$, and hence $\lambda_{\max}(\ell^2) \leq 1 + 2\sqrt{\frac14} = 2$.

    On the other hand, we do not know a particular~$\Gamma$ such that $\OurPsi(\Gamma) > \alpha(\Gamma)$ strictly.  Indeed, in \Cref{localopt} we consider the first and second derivatives of $\Opt(\ell)$, showing that the first derivative vanishes and the Hessian is negative semidefinite at the $\ell$ given in the proof  of \Cref{prop:alpha-lower}.
\end{remark}

\subsection{Implications for higher-degree polynomials}
The \lovasz Theta upper bound from \Cref{prop:lovasz-upper} can also be used to bound $\Opt_\pm(\p)$ for nonlinear polynomials~$\p$, simply by treating each monomial appearing in~$\p$ as a new indeterminate.
Let us work this out in the case of greatest interest to us: homogeneous polynomials over pairwise anticommuting indeterminates, as in \Cref{sec:KN}.
Suppose
\begin{equation} \label{eqn:homog2}
    \p = \i^{\binom{q}{2}} \sum_{\substack{S \subseteq [\N] \\ |S| = q}} a_S \X^S \in \CC(K_{\N}), \quad \text{with $\vec a \in \R^{\Nbig}$ satisfying $|\vec a| = 1$, where $\Nbig \coloneqq \tbinom{\N}{q}$}.
\end{equation}
Let us introduce $\Nbig$ new indeterminates $\ul{\X^S}$, where $\ul{\X^S}$ stands for $\i^{\binom{q}{2}} \X^S$; these are self-adjoint and square to~$1$.
Recalling from \Cref{fact:ac} how these monomials commute/anticommute leads us to define the following generalized Kneser-type graph (see the subsequent \Cref{def:KG} for a more general definition):
\begin{definition}
    Let $\Geven^{(\N,q)}$ denote the graph on vertex set $\binom{\N}{q}$ in which $\{S,T\}$ is a \emph{nonedge} if and only if the ``distance'' $d = q - |S \cap T|$ is even.
\end{definition}
We may now regard $\p$ from \Cref{eqn:homog2} as a self-adjoint, linear homogeneous element of $\CC(\Geven^{(\N,q)})$, and any upper bound on $\Opt_{\pm}(\p)$ in this setting is a valid upper bound on $\Opt_{\pm}(\p)$ in the original $\CC(K_{\N})$ setting.
(This is because there are only more relations among the actual monomials~$\X^S$, as compared to the new indeterminates~$\ul{\X^S}$.)
We may therefore conclude from \Cref{prop:lovasz-upper}:
\begin{proposition}                                     \label{prop:lovasz-kneser}
    For $\p$ as in \Cref{eqn:homog2} we have
    \[
        \Opt_{\pm}(\p)^2 = \Opt(\p^* \p) \leq \SOS_{2q}(\p^* \p) \leq \th(\Geven^{(\N,q)}).
    \]
\end{proposition}
This raises the question of determining $\th(\Geven^{(\N,q)})$.
As far as we are aware, this problem has not been previously studied.
The related quantity $\alpha(\Geven^{(\N,q)})$ \emph{is} well-understood, as a special case of work of Deza, Erd\H{o}s, and Frankl~\cite{DEF78} (discussed further in \Cref{sec:assoc}):
\begin{theorem}                                     \label{thm:DEF1}
    (\cite{DEF78}.)
    For $q$ even and all $\N$ sufficiently large, $\alpha(\Geven^{(\N,q)}) \leq \binom{\N/2}{q/2}$, and this is tight for even~$\N$, as
    $
        \braces*{I_1 \cup I_2 \cup \cdots \cup I_{q/2} : \text{distinct } I_j \in \{\{1,2\}, \{3,4\}, \dots, \{\N-1,\N\}\}}
    $
    is an independent set.

    Similarly, for $q$ odd and all $\N$ sufficiently large, $\alpha(\Geven^{(\N,q)}) \leq \binom{(\N-1)/2}{(q-1)/2}$,
    and this is tight for odd~$\N$, as
    $
        \braces*{I_1 \cup I_2 \cup \cdots \cup I_{(q-1)/2} \cup \{\N\}: \text{distinct } I_j \in \{\{1,2\}, \{3,4\}, \dots, \{\N-2,\N-1\}\}}
    $
    is an independent set.
\end{theorem}
We have strong evidence for the following conjecture (which implies \Cref{thm:DEF1}):
\begin{conjecture} \label{conj:lovasz}
    \Cref{thm:DEF1} holds also with $\th(\Geven^{(\N,q)})$ in place of $\alpha(\Geven^{(\N,q)})$.
\end{conjecture}
If true, this would imply $\Opt_{\pm}(\p) \leq O(\N^{\lfloor q/2\rfloor/2})$ for all $\p$ as in \Cref{eqn:homog2}.
Indeed, using the method of \Cref{sec:assoc} we can prove \Cref{conj:lovasz} for any particular small~$q$ (e.g., any $q \leq 10$), with computer assistance.
As $q = 4$ is of most interest to us, we explicitly verify the conjecture in this case in \Cref{sec:assoc2}, establishing:
\begin{theorem}                                     \label{thm:lovasz4}
    $\th(\Geven^{(\N,4)}) \leq \binom{\N/2}{2}$ for all $\N \geq 12$.
\end{theorem}
Combining this with \Cref{prop:lovasz-kneser}, we conclude:
\begin{corollary}                                     \label{cor:lovasz4}
    For any $\displaystyle \p = \sum_{|S| = 4} a_S \X^S \in \CC(K_\N)$ with $a \in \R^{\N}$ satisfying $|\vec a| = 1$, we have
    \[
        \Opt_{\pm}(\p) \leq \sqrt{\binom{\N/2}{2}} \sim \tfrac{1}{2\sqrt{2}} \N,
    \]
    and indeed $\SOS_8(\p^2) \leq \binom{\N/2}{2}$, for any $\N \geq 12$.

Remark: if there is an SOS proof that $\p^2 \leq \theta$ for some scalar $\theta$, there is also an SOS proof of the same degree that $\pm \p \leq \sqrt{\theta}$.
\end{corollary}
\begin{remark}  \label{rem:richardson}
    The bound in \Cref{cor:lovasz4} is exactly optimal for even~$\N$.
    To see this, consider
    \begin{align}
    \label{sqham}
        \p_1(\ulX)
        &= (\X_1 \X_2 + \X_3 \X_4 + \cdots + \X_{\N-1} \X_\N)^* (\X_1 \X_2 + \X_3 \X_4 + \cdots + \X_{\N-1} \X_\N) \\
        &= \N/2 - 2 \sum \braces*{\X^{I \cup J} : I,J \in \{\{1,2\}, \dots, \{\N-1, \N\}\},\  I \prec J}.
    \end{align}
    Thus
    \[
        \p = \frac{\p_1 - \N/2}{2\sqrt{\tbinom{\N/2}{2}}}
    \]
    is of the required form; i.e., it is degree-$4$ homogeneous with squared coefficients summing to~$1$.
    It now suffices to show that $\Opt(\p_1) = (\N/2)^2$, as this implies $\Opt(\p) = \sqrt{\binom{\N/2}{2}}$.
    Recalling \Cref{eg:fix}, it suffices to show that
    \[
        \|M\| = \N/2, \quad \text{where } M = -\i(\gamma_1 \gamma_2 + \gamma_3 \gamma_4 + \cdots + \gamma_{\N-1} \gamma_\N)
    \]
    and the $\gamma$-matrices are as in~\Cref{eqn:gamma}.
    But $M = Z \otimes \Id \otimes \cdots \otimes \Id + \Id \otimes Z \otimes \cdots \otimes \Id + \cdots + \Id \otimes \Id \otimes \cdots \otimes Z$, which is the $2^{\N/2 \times \N/2}$ diagonal matrix whose $x$th entry, for $x \in \{\pm 1\}^{\N/2}$, is $\sum_j x_j$.
    Hence indeed $\|M\| = \N/2$.
\end{remark}

\begin{remark}
\Cref{sqham} gives a $\p_1$ which is proportional to a square of a quadratic, and it is optimized by a Gaussian state, as discussed later.  In \Cref{lowrank} we show that that any $\p$ which is a weighted sum of squares of a small number of quadratics is approximately optimized by Gaussian states.
\end{remark}

\begin{remark}
    Restricting attention to even~$\N$, the lower bound $\N \geq 12$ in \Cref{cor:lovasz4} is best possible, since for $\N \in \{8, 10\}$ it even holds that $\alpha(\Geven^{(\N,4)}) = 14$.
    To see that $\alpha(\Geven^{(8,4)}) \geq 14$, consider the~$7$ lines of the Fano plane, though of as elements of $\binom{[7]}{3}$; then add the element~$8$ to each line, giving a $7$ subsets of $\binom{8}{4}$ where any two have intersection size exactly~$2$; finally, include also the~$7$ complements (within~$[8]$) of these subsets.
 In physics language, we can understand this example as follows: take terms
    $\X_1 \X_2 \X_3 \X_4$ and $\X_5 \X_6\X_7 \X_8$ with positive coefficients.  Restricting the $+1$ eigenspace of these two terms, the system may be described by $2$ qubits, with $X,Y,Z$ operators on the first qubit being bilinears in $\X_1,\X_2,\X_3,\X_4$ and  $X,Y,Z$ operators on the second qubit being bilinears in $\X_5,\X_6,\X_7,\X_8$.  The remaining terms in $h$ then are quartic operators whose sum is proportional to $X_1X_2+Y_1 Y_2+Z_1Z_2$.
\end{remark}

\subsection{On bounding $\th(\Geven^{(\N,q)})$} \label{sec:assoc}

\begin{notation}
    Recall the \emph{Johnson association scheme} $J(n,q)$ has $\binom{n}{q}$ as its points, with $S,T$ belonging to the $d$th relation, $d = 0 \dots q$, iff $\dist(S,T) \coloneqq q - |S \cap T| = d$.
    We write $A^{(n,q)}_d$ for the adjacency matrix of the $d$th relation, with the superscript dropped when it is understood from context.
    Note that $A_0 = \Id$ and $\sum_{d=0}^q A_d = \J$.
    For simplicity we assume $n > 2q$, so that $A_1, \dots, A_d$ represent connected graphs.
\end{notation}

\begin{definition}  \label{def:KG}
    Given $n$, $q$, and a set of \emph{nonedge distances} $\calD \subseteq \{1, \dots, q\}$, we write $\G_{\calD} = \G^{(n,q)}_{\calD}$ for the graph with vertex set~$\binom{n}{q}$ and adjacency matrix $\sum_{d \not \in \calD} A_d$.
\end{definition}

\begin{examples}
    The Erd\H{o}s--Ko--Rado theorem~\cite{EKR61} concerns the case $\calD  = \{1, 2, \dots, q-k\}$, for $k \in \NN^+$; we write the associated graph as $\Gekr{k}$.
    The theorem says that for sufficiently large~$n$ it holds that $\alpha(\Gekr{k}) = \binom{n-k}{q-k}$, with the maximum independent sets being of the form $\{S \subseteq [n] : S \supseteq T\}$ for some $|T| = k$.
\end{examples}

\begin{examples}
    We are particularly interested in the case when $\calD = \{2, 4, 6, \dots \}$ is the set of even integers; this yields the case of $\Geven = \Geven^{(n,q)}$ discussed in \Cref{thm:DEF1}.
\end{examples}

Deza, Erd\H{o}s, and Frankl~\cite{DEF78} more generally showed the following:
\begin{theorem} \label{thm:DEF}
    (\cite{DEF78}.)
    If the elements of~$\calD$ are $d_1 > d_2 > \cdots > d_r$, then $\alpha(\G_{\calD}) \leq \prod_{j = 1}^r \frac{n - q + d_j}{d_j}$ for sufficiently large~$n$.
\end{theorem}
Indeed, Deza--Erd\H{o}s--Frankl furthermore showed that unless the divisibility criterion $(d_1 - d_2) \mid (d_2 - d_3) \mid \cdots \mid (d_{r-1} - d_r) \mid d_r$ is satisfied, one has the asymptotically stronger bound~$\alpha(\G_{\calD}) \leq O_q(n^{r - 1})$.\\

Recall that for any graph at all we have $\alpha(G) \leq \th(G)$.
In the case of graphs coming from association schemes, the \lovasz Theta Function reduces to the ``linear programming bound'' of Delsarte~\cite{Del73}.
The next three facts review this bound in the context of the Johnson scheme.
\begin{proposition}                                     \label{prop:dualHahn}
    (\cite[p.~48]{Del73}.)
    The matrices $A_0, \dots, A_q$ are simultaneously diagonalizable, with eigenvalues given by the dual Hahn polynomials:
    \[
        \spec(A_d) = \{\tilde{H}_d(z) : z = 0 \dots q\}, \quad \tilde{H}_d(z) = \sum_{j=0}^d (-1)^{d-j}\binom{q-j}{d-j}\binom{q-z}{j}\binom{n-q+j-z}{j}.
    \]
    In particular, the all-$1$'s vector is an eigenvector of $A_d$ with eigenvalue $\tilde{H}_d(0) = \binom{n-q}{d} \binom{q}{d}$.
    This is an eigenvalue of multiplicity~$1$ for all $d > 0$.
\end{proposition}
\begin{lemma}                                       \label{lem:symmetrize}
    Fix any $\G = \G^{(n,q)}_{\calD}$.
    Then in the characterization of $\th(G)$ given in \Cref{eqn:dual}, it is equivalent to minimize only over~$\Sigma$ such that $\Sigma_{ST}$ depends only on~$\dist(S,T)$.
    That is,
    \[
        \th(\G) = \min \{\kappa : \exists (c_e)_{e \not \in \calD},\ \kappa(\Id + \littlesum_{e \not \in \calD} c_e A_e) \geq \J\}.
    \]
\end{lemma}
\begin{proof}
    Given any feasible~$\Sigma$ for \Cref{eqn:dual} achieving~$\kappa$, we may replace it by $\hat{\rho} = \avg_{\pi \in S_n} \{P(\pi)^\transp \Sigma P(\pi)\}$, where $P(\pi)$ is the permutation matrix on $\binom{n}{q}$ induced by~$\pi$.
    It is easy to see that $\hat{\rho}$ continues to be feasible and achieve~$\kappa$, while also having $\hat{\rho}_{ST}$ depend only on $\dist(S,T)$.
\end{proof}
\begin{corollary}                                       \label{cor:symmetrize}
    Fix any $\G = \G^{(n,q)}_{\calD}$.
    Then we have the following LP characterization:
    \[
        \th(\G) = \min_{(c_e)_{e \not\in \calD}} \braces*{\binom{n}{q}\Bigm/p(0) : p(1), \dots, p(q) \geq 0 \text{ where } p(z) = \sum_{e \not \in \calD} c_e \tilde{H}_e(z)}.
    \]
\end{corollary}
\begin{proof}
    This follows directly from \Cref{prop:dualHahn}.
\end{proof}

In the Erd\H{o}--Ko--Rado case, Schrijver~\cite{Sch81} was the first to show that in fact $\th(\Gekr{k}) = \alpha(\Gekr{k}) = \binom{n-k}{q-k}$ for all $n \geq n_0(q,k)$; Wilson~\cite{Wil84} later gave the optimal \mbox{$n_0(q,k) = (k+1)(q-k+1)$}.
Schrijver further asked whether \Cref{thm:DEF} could in general be strengthened by replacing $\alpha(\G_{\calD})$ with~$\th(\G_{\calD})$.
Our \Cref{conj:lovasz} is the special case of this question when $\calD = \{2, 4, 6, \dots\}$.
In the next section, we show how to verify the conjecture for~$q=4$; similar techniques can be used to verify the conjecture for any particular small~$q$.

\subsection{Proving \Cref{thm:lovasz4}}     \label{sec:assoc2}
We now explicitly prove \Cref{conj:lovasz} in the case of $q = 4$.

\begin{proof}[Proof of \Cref{thm:lovasz4}.]
    By \Cref{cor:symmetrize} we have
    \[
            \th(\Geven^{(n,4)}) = \min_{c_1, c_3} \braces*{\binom{n}{4}\Bigm/p(0) : p(1), p(2), p(3), p(4) \geq 0 \text{ where } p(z) = c_1 \tilde{H}_1(z) + c_3 \tilde{H}_3(z)}.
    \]
    The minimizer will occur when two of $p(1), p(2), p(3), p(4)$ are~$0$.
    We first consider fixing $p(4) = 0$.
    It is not hard to compute that $p(4) = 1 - 4c_1 - 4c_3$, and hence $p(4) = 0$ provided
    \[
        c_1 = \frac14 - c_3.
    \]
    Fixing this choice, and also writing $n' = n-4$, we may compute:
    \begin{align*}
        p(0) &= 1 + 4n'(\tfrac14 - c_3) + 4\tbinom{n'}{3}c_3 \\
        p(1) &= 1 + (3n' - 4)(\tfrac14 - c_3) + \parens*{\tbinom{n'-1}{3}-3\tbinom{n'-1}{2}}c_3 \\
        p(2) &= 1 + (2m'-6)(\tfrac14 - c_3) - \parens*{2\tbinom{n'-2}{2}-2(n'-2)}c_3\\
        p(3) &= 1 + (m'-6)(\tfrac14 - c_3) - (3(n'-3)-1)c_3.
    \end{align*}
    We now consider fixing $p(2) = 0$ by choosing
    \[
        c_3 = \frac{1}{2(n'-4)}.
    \]
    This yields
    \[
        p(0) = \frac{(n-1)(n-3)}{3}, \quad p(1) = \frac{(n-2)(n-4)(n-12)}{12(n-8)},  \quad p(3) = \frac{(n-4)(n-6)}{4(n-8)}.
    \]
    Provided $n \geq 12$, we get $p(1), p(3) \geq 0$ and thereby establish
    \[
        \th(\Geven^{(n,4)}) \leq \binom{n}{4} \Bigm/ \parens*{\frac{(n-1)(n-3)}{3}} = \binom{n/2}{2}. \qedhere
    \]
\end{proof}

\section{Efficient upper bound certification for degree-$4$ SYK in degree-$8$ Sum-of-Squares}   \label{sec:refutingSYK4}
\newcommand{\etas}{\tau}
\newcommand{\boldetas}{\btau}
Let us consider the algebra over indeterminates $\X_1, \dots, \X_\N$ and $\etas_1, \dots, \etas_\N$ in which we initially assume only the following relations:
\begin{equation} \label{eqn:etas-axioms}
    \X_j^* = \X_j, \quad \etas_j^* = \etas_j, \quad \X_j^2 = 1 \quad \forall j \in [\N].
\end{equation}
In particular, we do not necessarily assume that $\etas_j \etas_k = -\etas_k \etas_j$ or $\etas_j^2 = 1$.
(The anti/commutation of the $\X_j$'s will be irrelevant.)
Consider further the self-adjoint polynomial
\[
    \p = \frac{\i}{\sqrt{\N}} \sum_{m=1}^\N \etas_m \X_m.
\]
By the matrix form of Cauchy--Schwarz we can immediately conclude that
\[
    \Opt_{\pm}(\p) \leq \sqrt{\Opt\parens*{\sum_{m=1}^\N \etas_m^2}} \sqrt{\Opt\parens*{\tfrac{1}{\N} \sum_{m=1}^\N \X_m^2}} = \sqrt{\Opt\parens*{\sum_{m=1}^\N \etas_m^2}},
\]
but let us deduce this slightly differently for the sake of SOS:
\begin{proposition}                                     \label{prop:etas}
    Let $\beta > 0$, and suppose there is a degree-$k$ ($k \geq 2$) SOS proof of $\sum_{m=1}^\N \etas_m^2 \leq \beta^2$.
    Then there is also one of $\pm \p \leq \beta$.
\end{proposition}
\begin{proof}
    The inequalities are obtained by choosing $\alpha = \sqrt{\beta/\sqrt{\N}}$ in
\[
    \tfrac12(\i \alpha^{-1} \etas_m \pm  \alpha \X_m)^*(\i \alpha^{-1} \etas_m \pm \alpha\X_m) \geq 0 \quad\implies\quad \mp \i \etas_m \X_m \leq \tfrac{1}{2\alpha^2} \etas_m^2 + \tfrac{\alpha^2}{2} \X_m^2,
\]
and summing over~$m$.
\end{proof}
\newcommand{\g}{g}
\begin{lemma}                                       \label{lem:deg3squared}
    For $\vec a \in \R^{\binom{m}{3}}$, let
    $\displaystyle
        \g = \i \sum_{\substack{S \subseteq [m] \\ |S| = 3}} a_S \X^S \in \CC(K_m).
    $
    Then
    \[
        \g^2 = |\vec a|^2 + \sum_{\substack{U \subseteq [m] \\ |U| = 4}} \parens*{\sum_{\ell \in [m] \setminus U}  \sum_{\substack{\{S,T\} \colon U = S \cup T \\ |S| = |T| = 2}} \pm 2 a_{S \cup \{\ell\}}  a_{T \cup \{\ell\}}} \X^U
    \]
    for certain signs~$\pm$.
\end{lemma}
\begin{proof}
    It is clear that $\g^2$ will have terms only of even degree.
    Further, since it is self-adjoint, these terms can only be of degree-$0$ or degree~$4$ (cf.~\Cref{fact:rev}).
    The lemma now follows by explicit computation.
\end{proof}

We now introduce some matrix and tensor notation.
Let $\p=\sum_{\substack{S \subseteq [\N] \\ |S| = 4}} J_S \X^S$ for some
given choice of $J_S$.
Define a $4$-index tensor $J_{ijkl}$ with entries $J_{ijkl}=J_{\{i,j,k,l\}}$ for $i<j<k<l$ and for other orderings of $i,j,k,l$ let $\pi$ be a permutation on $4$ elements so that $\pi(i)<\pi(j)<\pi(k)<\pi(l)$ and then $J_{ijkl}=J_{\pi(i),\pi(j),\pi(k),\pi(l)} \mathrm{sign}(\pi)$.
Conversely, given an arbitrary tensor $J_{ijkl}$ define
$\p(J)\equiv \sum_{i,j,k,l\quad \mathrm{distinct}} J_{ijkl} \X_i \X_j \X_k \X_l$, where the sum is over \emph{distinct} $i,j,k,l$.
(If $J$ is completely anti-symmetric then there is no need to restrict to
distinct $i,j,k,l$, but for general $J$ this restriction matters.)
From this tensor $J_{ijkl}$ define an $\N^2$-by-$\N^2$ matrix $J^{\mathrm{mat}}$ by regarding $i,j$ as row indices and $k,l$ as column indices.

Having introduced this notation we have:
\begin{theorem}
\label{J4bound}
There is a constant $C>0$ such that
for any $\p=\sum_{\substack{S \subseteq [\N] \\ |S| = 4}} J_S \X^S$
we have
$\SOS_8(\p)\leq C \N \Vert J^{\mathrm{mat}} \Vert_4$,
where $\Vert \ldots \Vert_4$ denotes the Schatten $4$-norm.
\begin{proof}
Throughout this proof we use $C',C'',\ldots$ to denote arbitrary constants.
Write
$\p=\frac{1}{\sqrt{\N}}\sum_i \X_i \tau_i$, where we define
$\tau_i=\frac{\sqrt{\N}}{4} [\p,\X_i]$, i.e., up to constant factors $\tau_i=\sqrt{\N} \sum_{j,k,l} J_{ijkl} \X_j \X_k \X_l$.  Remark: the reason for multiplying $\tau_i$ by $\sqrt{\N}$ and then dividing by $\sqrt{\N}$ is that it is a more convenient normalization for random instances.

We will show that there is a degree-$8$ proof that
$\sum_{m=1}^\N \tau_m^2 \leq C' \N^2 \Vert (J^{\mathrm{mat}})^2\Vert_2$.
Then by \Cref{prop:etas}, there is a degree-$8$ proof that
$\p\leq C \N \Vert J^{\mathrm{mat}}\Vert_4$.

By \Cref{lem:deg3squared}, $\tau_m^2$ is a sum of degree-$0$ and degree-$4$
terms.
In fact, we may express these terms conveniently using matrix notation.
We claim that
\be
\label{4indexsum}
\sum_m \tau_m^2=\frac{\cdot 3!}{(4!)^2} \N |J|^2+\frac{3!}{(4!)^2} \N
\p((J^{\mathrm{mat}})^2),
\ee
where $\p((J^{\mathrm{mat}})^2)$ is defined by squaring $J^{\mathrm{mat}}$, regarding it as a matrix, then regarding the resulting matrix as a $4$-index tensor as an input to the function $\p(\cdot)$.

By,
\Cref{cor:lovasz4}, $\p((J^{\mathrm{mat}})^2)$
is bounded by
$C' \N |(J^{\mathrm{mat}})^2|$.
Remark: to prove this in some more detail, take $(J^{\mathrm{mat}})^2$, regard it as a $4$-index tensor, then project onto the subspace which is fully anti-symmetric of permutations of indices, and then apply \Cref{cor:lovasz4}.  However, this projection cannot increase the $\ell_2$ norm.

Thus,
$\SOS_8(\p)\leq \sqrt{C'\N^2 |(J^{\mathrm{mat}})^2|+C'' \N |J|^2}$.
Since $\N |J|^2\leq \N^2 |(J^{\mathrm{mat}})|^2$, the theorem follows.

To verify \Cref{4indexsum},
we have
$$\sum_m \tau_m^2=\frac{n}{(4!)^2} \sum_m \sum_{jkl}\sum_{opq} J_{mjkl} J_{mopq} \X_j \X_k \X_l \X_o \X_p \X_q.$$
We break this sum up into several cases, depending on the number of indices $j,k,l$ which match indices $o,p,q$, i.e., on $|\{j,k,l\} \cap \{o,p,q\}|$.
If $0$ or $2$ indices match, then the resulting term is degree $6$ or $2$ and these terms vanish.
So, we may consider the case that $1$ or $3$ indices match.
The sum of terms with $3$ indices matching is equal to
$$\frac{3!}{(4!)^2} n|J|^2,$$
where to derive we may pick an arbitrary permutation for indices matching (i.e., assume $j=o, k=p, l=q$) and then sum over the $3!$ permutations.

The sum of terms with $1$ index matching is equal to
$$\frac{9}{(4!)^2} n \sum_{m,j} \sum_{k,l,p,q \, \mathrm{distinct}}
J_{mjkl} J_{mjpq} \X_k \X_l \X_p \X_q,$$
where again we picked an arbitrary choice of indices matching (here picking $j=o$) and summed over the $9$ such choices.
\end{proof}
\end{theorem}

\Cref{J4bound} has the immediate corollary that
$\E[\SOS_8(\bh)]=O(\sqrt{\N})$ as one may estimate $\Vert J^{\mathrm{mat}} \Vert_4$.
However, we can tighten the constants using a slightly different choice of variables $\tau$ as in the following theorem:
\begin{theorem}                                     \label{thm:syk4sos}
    Let $\bh \sim \SYK_4(\N)$, where $\N \geq 12$; i.e., for independent standard Gaussians $\bJ_S$,
    \[
        \bh = \frac{1}{\sqrt{\binom{\N}{4}}} \sum_{\substack{S \subseteq [\N] \\ |S| = 4}} \bJ_S \X^S.
    \]
    Then $\E[\SOS_8(\bh)] \leq \sqrt{1+\sqrt{6}} \sqrt{\N}$.
\end{theorem}
\begin{proof}
    Write $\bh = -\frac{\i}{\sqrt{\N}} \sum_{m = 1}^\N \boldetas_m \X_m$, where
    \[
        \boldetas_m = \frac{2\i}{\sqrt{\binom{\N-1}{3}}} \sum_{\substack{T \subseteq [m-1] \\ |T| = 3}} \bJ_{T \cup \{m\}} \X^T.
    \]
    We remark that, up to normalization, the $\boldetas_m$'s are independently distributed as~$\SYK_3(m-1)$.
    Applying \Cref{lem:deg3squared} to $\boldetas_m$, we get
    \begin{align}
    \label{SOSfknowsthis}
        \sum_{m=1}^\N \boldetas_m^2 = \frac{4} {\binom{\N-1}{3}}
        \sum_{\substack{S \subseteq [\N] \\ |S| = 4}} \bJ_S^2
        + \sum_{m = 1}^\N
            \sum_{\substack{U \subseteq [m-1] \\ |U| = 4}} \bc_{m,U}\X^U
        = \frac{4} {\binom{\N-1}{3}}\sum_{\substack{S \subseteq [\N] \\ |S| = 4}} \bJ_S^2
        +  \sum_{\substack{U \subseteq [\N] \\ |U| = 4}}
            \sum_{m > \max(U)} \bc_{m,U}\X^U,
    \end{align}
    where
    \begin{align}
    \label{tensornet}
        \bc_{m,U} \coloneqq \frac{8}{\binom{\N-1}{3}} \sum_{\ell \in [m-1] \setminus U}  \sum_{\substack{\{S,T\} \colon U = S \cup T \\ |S| = |T| = 2}}
            \pm \bJ_{S \cup \{\ell,m\}}  \bJ_{T \cup \{\ell,m\}}.
    \end{align}
    Thus using \Cref{cor:lovasz4} we get
    \begin{align}
        \E\bracks*{\SOS_8\parens*{\sum_{m=1}^\N \boldetas_m^2}} &\leq
            \frac{4} {\binom{\N-1}{3}}\E\sum_{\substack{S \subseteq [\N] \\ |S| = 4}} \bJ_S^2 +
            \E\sqrt{\sum_{\substack{U \subseteq [\N] \\ |U| = 4}} \parens*{
            \sum_{m > \max(U)} \bc_{m,U}}^2} \sqrt{\binom{\N/2}{2}} \nonumber\\
            &\leq n +  \sqrt{\E\sum_{\substack{U \subseteq [\N] \\ |U| = 4}} \parens*{
            \sum_{m > \max(U)} \bc_{m,U}}^2} \sqrt{\binom{\N/2}{2}}. \label{eqn:bd}
    \end{align}
    To bound the expectation in \Cref{eqn:bd}, note that for a fixed~$U$, the random variables $\bc_{m,U}$ have mean zero and are independent across~$m$ (since $\bc_{m,U}$ depends only on $\bJ_R$ for $R \ni m$).
    Thus
    \[
        \E\sum_{\substack{U \subseteq [\N] \\ |U| = 4}} \parens*{
            \sum_{m > \max(U)} \bc_{m,U}}^2 =
        \sum_{\substack{U \subseteq [\N] \\ |U| = 4}} \sum_{m > \max(U)} \E[\bc_{m,U}^2] =
        \sum_{m=1}^\N \sum_{\substack{U \subseteq [m-1] \\ |U| = 4}} \E[\bc_{m,U}^2].
    \]
    Further, for fixed $m$ and $U$, the summands $\bJ_{S \cup \{\ell,m\}}  \bJ_{T \cup \{\ell,m\}}$ in $\bc_{m,U}$ are products of \emph{distinct} pairs of standard Gaussians, and hence
    \[
        \E[\bc_{m,U}^2] = \frac{64}{\binom{\N-1}{3}^2} \cdot (m-5) \cdot 3
    \]
    (and we see it did not matter what the $\pm$ signs were).
    Thus the expectation in \Cref{eqn:bd} is
    \[  =
        \sum_{m=1}^{\N} \binom{m-1}{4} \cdot \frac{64}{\binom{\N-1}{3}^2} \cdot (m-5) \cdot 3 = \frac{960 \binom{\N}{6}}{\binom{\N-1}{3}^2},
    \]
    and we conclude that
    \[
        \E\bracks*{\SOS_8\parens*{\sum_{m=1}^\N \boldetas_m^2}} \leq \N + \sqrt{\frac{960 \binom{\N}{6}}{\binom{\N-1}{3}^2}} \sqrt{\binom{\N/2}{2}} = \N + \sqrt{6\N^2 \cdot \frac{(\N-4)(\N-5)}{(\N-1)(\N-3)}} \leq (1+\sqrt{6})\N.
    \]
    The result now follows from \Cref{prop:etas}.
\end{proof}
\begin{remark}
    The factor $\sqrt{1 + \sqrt{6}} \approx 1.86$ in this theorem may be compared to the provable factor of~$\sqrt{\ln 2} \approx .83$ of~\cite{FTW19} (\Cref{rem:FTW}) and the conjectural factor of~$\frac{1}{2\sqrt{2}} \approx .35$ (see \Cref{eqn:heur}).
\end{remark}

\section{Lower degree Sum-of-Squares: certification at degree $6$ and fooling at degree $4$}
We have shown that, with high probability, degree-$8$ sum-of-squares certifies an upper bound to the largest eigenvalue of $\p$ for $\SYK_4$ which is within a constant factor of $\sqrt{\N}$.
In this section we consider lower degree sum-of-squares.

We show that also a similar statement holds for degree-$6$ sum-of-squares: with high probability, it also certifies an upper bound to the largest eigenvalue of $\p$ for $\SYK_4$ which is within a constant factor of $\sqrt{\N}$, albeit with a weaker constant.  Indeed, a certain ``fragment" of
degree-$6$ sum-of-squares is able to prove this, and this ``fragment" requires
only a constant factor increase in matrix size over degree-$4$.

Let us explain the reason for the interest in this ``fragment".  The interest is \emph{not} to give a faster certifier; indeed, \Cref{J4bound} requires only evaluating a certain Schatten $4$-norm which can be done faster than one can solve this ``fragment".  Rather, the interest is that we hope that this ``fragment" will be useful in general, not just for $\SYK_4$.  In addition to certifying $\SYK_4$, it can also prove everything that $\SOS_4$ can.  In contrast, the Schatten $4$-norm bound on its own likely does not give any particularly useful bound for typical problems of interest in quantum chemistry.

On the other hand, we will show that degree-$4$ sum-of-squares is not able to certify an $O(\sqrt{\N})$ upper bound.

\subsection{Degree-$6$ Sum-of-Squares for $\SYK_4$}
Let us
define a certain ``fragment" of $\SOS_6$.  We call this a ``fragment" because
it is defined by a positive semidefinite matrix that includes sum, but not all, of the information in $\SOS_6$.

We call this fragment $\SOSf$; the reason for this notation will become clear.
Briefly,
$\SOSf$ is defined by considering expectation values
of monomials up to degree $4$ in variables $\X$, as well as expectation variables of monomials
of degree $1$ in $\X$ and degree $1$ in $\tau$, and finally expectation
values of monomials of degree $2$ in $\tau$, subject to appropriate anticommutation constraints.  The notation $4,2$ in $\SOSf$ is intended to denote this total degree at most $4$ in $\X$ and at most $2$ in $\tau$.
Note that all expectation values of odd total degree in $\tau,\X$ may be assumed to vanish.

We can also formally repeat our previous definitions with slight modification:
\begin{definition}
    We say that $\p$ is a \emph{degree-$4,2$ sum of (Hermitian-)squares} (abbreviated \emph{SOS}) if it is expressible as $g_1^* g_1 + \cdots + g_m^* g_m$ for some polynomials $g_1, \dots, g_m$ of degree at most~$2$ in $\X$ and at most $1$ in $\tau$.
\end{definition}

\begin{definition}
    We define a \emph{degree-$4,2$ pseudoexpectation} to be a linear functional $\pE[\cdot] : \CC(\Gamma) \to \C$ with $\pE[1] = 1$ and $\pE[\p] \geq 0$ for all degree-$4,2$ SOS~$\p$.
    We therefore call a self-adjoint $\rho \in \CC(\Gamma)$ a \emph{degree-$4,2$ pseudostate} if has $\tr(\rho) = 1$ and $\tr(\rho \p) \geq 0$ for all degree-$4,2$ SOS~$\p$.
\end{definition}

We characterize this by the following positive semidefinite matrix:
\begin{fact}
\label{M42}
   Let $\pE[\cdot]$ be a linear functional on $\CC(\Gamma)$ with $\pE[1] = 1$.
Let $M$ be a matrix having rows and columns each indexed by a pair of sets $(S_\X,\S_\tau)$, with $S_\X,S_\tau \subset [\N]$, and with
$|S_\X|\leq 2$ and $|S_\tau|\leq 1$.
Let $M$ have
$((S_\X,S_\tau),(T_\X,T_\tau))$-entry equal to$\pE[(\X^{S_\X})^* (\tau^{S_\tau})^* \X^T \tau^{T_\tau}]$,
where $\tau^{T_\tau}$ is defined in the obvious way as the product of $\tau_i$ for $i\in T_\tau$.
    Then $\pE[\cdot]$ is a degree-$4,2$ pseudoexpectation if and only if~$M$ is PSD.
\end{fact}

\begin{fact}
\label{lincon}
The space of $M$ for arbitrary $\pE[\cdot]$ is subject to certain linear constraints.

We will write $M_{s,t}$ to denote certain submatrices of total degree at most $s$ in $\X$ and $t$ in $\tau$.

First, the submatrix $M_{4,0}$ indexed by rows and columns where $S_\tau=T_\tau=\emptyset$ is subject to the same linear constraints as for degree-$4$ pseudoexpectations: antisymmetry under permutation of indices, and constraints imposed by $\X^2=1$.

Second,
the entries of $M_{1,1}$
 are linearly determined by entries of $M_{4,0}$.  Indeed, the pseudoexpectation
 of $\X \tau_j$ for any $i,j$ is a linear function of $M_{4,0}$.

Third, the entries of $M_{0,2}$ which are
pseudoexpectations of $\tau_i \tau_j$ are subject to a constraint that
$\pE[\tau_i \tau_j + \tau_i \tau_j]$ is linear determined by entries of $M_{4,0}$.  Indeed, $\{\tau_i,\tau_j\}$ is degree at most $4$ in $\X$.
\end{fact}

\begin{fact}
\label{isfragment}
Every degree-$6$ pseudoexpectation is a degree-$4,2$ pseudoexpectation.  This follows since every degree-$4,2$ SOS is a degree-$6$ SOS since $\tau$ has degree $3$.

Similarly, given any degree-$6$ pseudoexpectation, and corresponding PSD matrix of expectation values $M$, the matrix of expectation values for the corresponding degree-$4,2$ pseudoexpectation is obtained by projecting $M$ into an appropriate subspace.
\end{fact}

\begin{theorem} \label{thm:fool}
    Let $\bh \sim \SYK_4(\N)$.
    Then, $\SOS_6(\bh)\leq \SOSf(\bh)$ and with
    high probability $\SOSf(\bh) \leq O(\sqrt{\N})$.

\begin{proof}
The fact that $\SOS_6(\bh)\leq \SOSf(\bh)$ follows from \Cref{isfragment}.

The proof that
$\E[\SOSf(\bh)] \leq O(\sqrt{\N})$ is similar to that of  \Cref{J4bound}, and we use the same definition for $\tau_m$.
The only new thing required is to show that with high probability
$\SOSf[\sum_{m=1}^\N \tau_m^2]=O(\N)$.

However, the linear constraints \Cref{lincon} show that
$\SOSf$ still knows \Cref{4indexsum}.  What we need to show is that
$\SOSf$ can bound the right-hand side of this equation.

Indeed, we will show that $\SOS_4$ can bound this.  Of course, $\SOS_4$ can bound the degree-$0$ terms.

Regard the space of $N^2$-by-$N^2$ matrices as a vector space, and label rows of such matrices by a pair $i,j\in [\N]$ and label columns by a pair $k,l \in [\N]$.  Then, define a projector $\pi_{\mathrm{distinct}}(\cdot)$ which projects onto the subspace where $i,j,k,l$ are all distinct.
The pseudo-expectation of the sum of degree-$4$ terms,
which is proportional to $\N\p((J^{\mathrm{mat}})^2)$, equals (in a slight abuse of notation)
$$\Tr(\N \pi_{\mathrm{distinct}}((J^{\mathrm{mat}})^2) M_{4,0}).$$
The abuse of notation is that a matrix such as $J^{\mathrm{mat}}$ is size $\N^2$-by-$\N^2$ (though it vanishes when $i=j$ or $k=l$), while $M_{4,0}$ is size  $n\choose{2}$-by-$n\choose{2}$, so one has essentially twice as many rows and columns as the other.  However, we will freely take the trace of one kind of matrix with the other by extending $M_{4,0}$ to a matrix of size $n^2$-by-$n^2$ in the obvious way using anti-symmetry, i.e.,if $M_{4,0}$ has some given coefficient $(M_{4,0})_{S,T}$ for $|S|=|T|=2$, then we extend it to a matrix with coefficients $(M_{4,0})_{i_1,i_2;j_1,j_2}$ which equals $\pm(M_{4,0})_{\{i_1,i_2\},\{j_1,j_2\}}$ where the sign depends on the permutation needed to make $i_1<i_2$ and $j_1<j_2$.

However, 
$$
\Tr(\N \pi_{\mathrm{distinct}}((J^{\mathrm{mat}})^2) M_{4,0})\leq
\Tr(\N (J^{\mathrm{mat}})^2 M_{4,0}),$$ as the terms where $i,j=k,l$ or $i,j=l,k$ contribute positively.

 Further, $M_{4,0}$ is a PSD matrix with trace $O(\N^2)$.  Hence, this pseudoexpectation is bounded by
 $O(\N^2) \| \N (J^{\mathrm{mat}})^2 \|_{\mathrm{op}}=O(\N^3) \| J^{\mathrm{mat}} \|_{\mathrm{op}}^2$.
 Standard random matrix theory
  results show that with high probability $\| J^{\mathrm{mat}} \|_{\mathrm{op}}=O(1/\N)$.  This completes the proof\footnote{To give the bound
 on $\| J^{\mathrm{mat}} \|_{\mathrm{op}}$ in more detail, such a matrix is a matrix with random Gaussian entries, which are independent up to the anti-symmetry.  We can draw matrices from that distribution by taking a random Gaussian matrix with \emph{independent} entries (and the same variance up to constant factor) and summing over signed permutations of the indices to project onto the totally antisymmetric subspace.  Then, using standard random matrix bounds on the operator norm of the matrix with independent entries and a triangle inequality gives us the desired bound on  $\| J^{mat} \|{\mathrm{op}}$.  We use similar bounds elsewhere and \cite{HTS21} uses the same bound.}.

Remark: it is interesting to compare this degree-$4$ bound on $z$, namely by
$O(\N^3) \| J^{\mathrm{mat}} \|_{\mathrm{op}}^2$, to the bound we used in the proof of \Cref{J4bound}, namely bounding $z$ by
$O(\N^2) \Vert (J^{\mathrm{mat}})^2 \Vert_2$. The bound used in \Cref{J4bound} is tighter, but they have the same scaling for a matrix such as $J^{\mathrm{mat}}$ where ``most" eigenvalues are of order $O(1/\N)$.
\end{proof}
\end{theorem}

\subsection{Fooling Degree $q$}
\begin{theorem}
  \label{thm:fools4}
  Let $\p=\sum_S $ be an instance of $\SYK_q$ for $q$ even.
  Let $\rho=1+C \p$, where $C$ is a scalar.
  Then with high probability $\tr(\rho \p)=C\cdot (1+o(1))$.
  Further,
 if $|C|$ is sufficiently small compared to $\N^{q/4}$, then with high probability, $\rho$ is a degree-$q$
  pseudostate.

  Consequently, for $\bh \sim \SYK_4$, with high probability
  $\SOS_4(\bh)=\Omega(\N)$.
  \begin{proof}
  Let $\vec a$ be the vector of coefficients of terms in $\p$.
  Then, $\tr(\rho \p)=C |\vec a|^2$ and with high probability $|\vec a|^2=1+o(1)$.

To verify that it is a degree-$q$-pseudostate, we must verify that $\tr(g\rho g^*)\geq 0$ for all $g$ of degree at most $q/2$.  Note that $\tr(\rho m)$ vanishes for any $m$ which is homogeneous of degree not equal to $0$ or $q$.
Hence, it suffices to verify that the inequality holds when $g$ is homogeneous of degree $0$ (which is immediate) and when $g$ is homogeneous of degree $q/2$.

In case that $g$ is homogeneous of degree-$q/2$,
$$g= \sum_{\substack{S \subseteq [\N] \\ |S| = q/2}} b_S \X^S,$$
we have
\be
\label{eigbound}
\tr(g\rho g^*)=\sum_{\substack{S \subseteq [\N] \\ |S| = q/2}} |b_S|^2+
C\tr(g\p g^*).
\ee
To evaluate
$\tr(g\p g^*)$,
we give a representation in terms of matrices (essentially the same representation is used in \cite{HTS21}).

Define a matrix $m$ whose rows and columns are indexed by subsets of $[n]$ with cardinality $q/2$.  The coefficient $m_{S,T}$ (where $S,T\subset [n]$ and $|S|=|T|=n/2$) is defined to vanish unless $S\cup T=\emptyset$, in which case  $m_{S,T}=\tr(\X^S h \X^T)=\pm a_{S\cup T}$.  The sign may be worked out using the anticommutation relations, of course.
Then,
$\tr(g\p g^*)=\langle b | m | b \rangle$.

However, standard random matrix theory bounds show that with high probability $\| m \|_{\mathrm{op}} = O(\N^{-q/4})$ and so
  indeed
 if $|C|$ is sufficiently small compared to $\N^{q/4}$, then with high probability, $\rho$ is a degree-$q$
  pseudostate.
\end{proof}
\end{theorem}

\section{Gaussian states}   \label{sec:gaussian-states}

This section will focus exclusively on $\CC(K_{\N})$, where all indeterminates pairwise anticommute.
We will discuss \emph{Gaussian states}~\cite{TD02,Bra05}, a reasonably broad class of states~$\rho$ (including all Slater determinant states) for which: (i)~there is an efficient representation of~$\rho$; (ii)~given this representation, and some polynomial~$\p \in \CC(K_{\N})$, there is an efficient procedure for computing $\E_\rho[\p]$.
In particular, these states may be used for classical (implicit) lower-bound witnessing algorithms, as described in \Cref{sec:complexity}.

A succinct characterization of Gaussian states is that they are precisely those states that are optimizers for homogeneous degree-$2$ Hamiltonians~\cite{Bra05}.
To define them more explicitly, first suppose $\psi_1, \dots, \psi_m \in \CC(\Gamma)$  are self-adjoint, commuting involutions ($\psi_j^2 = 1$).
Then for any $\lambda_1, \dots, \lambda_m \in \{\pm 1\}$, the elements $1+\lambda_j \psi_j$ are positive and commuting, and hence
\[
    \rho = \prod_{j=1}^m (1+\lambda_j \psi_j)
\]
is a state (note that the order of the factors is irrelevant).
Indeed, this was already pointed out in \Cref{eg:pmstate}, just with different letters.
Furthermore, the same is true if we generalize by allowing $\lambda_1, \dots, \lambda_m \in [-1,1]$.
Additionally, if $\Gamma = K_\N$ with~$\N$ even, we may take $\psi_1 = \i\X_1 \X_2$, $\psi_2 = \i\X_3 \X_4$, etc.
Even more generally, we may augment this construction with the orthogonal transformation described in \Cref{fact:ortho}.
These generalizations lead us to the class of Gaussian states:
\begin{definition}
    Consider the setting of $\CC(K_\N)$ with $\N$ even.
    Let $\lambda_1, \dots, \lambda_{\N/2} \in [-1,1]$, let $O \in \R^{\N \times \N}$, and let $\ell = O\X$.
    Then any element
    \begin{equation}    \label{eqn:gs}
        \rho = \prod_{j=1}^{\N/2} (1 + \i \lambda_j \ell_{2j-1} \ell_{2j}) \in \CC(K_\N)
    \end{equation}
    is called a \emph{Gaussian state}.
\end{definition}
\begin{examples} \label{eg:ket0}
    For example, if $O = \Id$ and $\lambda_j = -1$ for each~$j$, we have the Gaussian state
    \[
        \prod_{j=1}^{\N/2} (1 - \i \X_{2j-1} \X_{2j}).
    \]
    Under the canonical representation using the $\gamma$-matrices from \Cref{eqn:gamma}, this is
    \[
        \prod_{j=1}^{\N/2}(\Id +  \underbrace{Z \otimes \cdots \otimes Z}_{j \text{ times}} \otimes \underbrace{\Id \otimes \cdots \otimes \Id}_{\N/2 - j \text{ times}}) = 2^{\N/2} \ket{00\cdots 0}\!\bra{00\cdots 0}.
    \]
\end{examples}

We will also be interested in the case when the condition $|\lambda_j| \leq 1$ is not enforced:
\begin{definition}
    We call $\rho \in \CC(K_\N)$ a \emph{Gaussian pseudostate} if it has the form \Cref{eqn:gs}, but with any $\lambda_1, \dots, \lambda_{\N/2} \in \R$ allowed.
    In general we will write $\pE_\rho[\p] = \tr(\rho \p)$, with $\pE_\rho[\cdot]$ being at worst a degree-$0$ pseudoexpectation.
    (If $\rho$ is known to be a proper state we will switch to the notation~$\pE_\rho[\cdot]$.)
\end{definition}

\begin{definition}
    A Gaussian pseudostate as in \Cref{eqn:gs} is characterized by its \emph{covariance matrix} $\Sig \in \R^{\N \times \N}$, which is the real antisymmetric matrix defined by $\Sig = O^\transp \Lambda O$, where $\Lambda$ is block-diagonal with $2 \times 2$ blocks of the form $\begin{pmatrix} 0 & \lambda_j \\ -\lambda_j & 0 \end{pmatrix}$; it satisfies $\Sig_{jk} = \pE_\rho[\i \X_j \X_k]$ for all $1 \leq j,k \leq \N$.
    Conversely, any real antisymmetric $\Sig$ defines a Gaussian pseudostate which we will denote by~$\rho_\Sig$; this is a proper Gaussian state provided $\|\Sig\|_{\mathrm{op}} \leq 1$.
\end{definition}

\begin{definition}
    We introduce the terminology \emph{Gaussian pseudostate} for the case when $|\lambda_j| \leq 1$ is not necessarily enforced.
    A Gaussian pseudostate is characterized by its \emph{covariance matrix} $\Sig \in \R^{\N \times \N}$, which is the real antisymmetric matrix defined by $\Sig = O^\transp \Lambda O$, where $\Lambda$ is block-diagonal with $2 \times 2$ blocks of the form $\begin{pmatrix} 0 & \lambda_j \\ -\lambda_j & 0 \end{pmatrix}$; it satisfies $\Sig_{jk} = \pE_\rho[\i X_j X_k]$ for all $1 \leq j,k \leq \N$.
    Conversely, any real antisymmetric $\Sig$ defines a Gaussian pseudostate which we will denote by~$\rho_\Sig$; this is a proper Gaussian state provided $\|\Sig\|_{\mathrm{op}} \leq 1$.
\end{definition}
\begin{fact}                                        \label{fact:wick4}
    For the Gaussian pseudostate $\rho_\Sig$, and $S \subseteq \binom{[\N]}{4}$, we have
    \[
        {\pE}_{\rho_\Sig}[\X^S] = -\parens*{\Sig_{S_1 S_2} \Sig_{S_3 S_4} - \Sig_{S_1 S_3} \Sig_{S_2 S_4} + \Sig_{S_1 S_4} \Sig_{S_2 S_3}},
    \]
    as is easily verified.
    More generally, the pseudomoments may be computed by the fermionic Wick formula: for $S \subseteq [\N]$ we have
    $\E_{\rho_\Sig}[\X^S] = \i^{\binom{|S|}{2}} \mathrm{Pf}(\Sig_{S,S})$,
    where $\mathrm{Pf}$ denotes the Pfaffian and $\Sig_{S,S}$ is the submatrix of~$\Sig$ with rows/columns indexed by~$S$.
\end{fact}

\begin{example} \label{eg:richardson}
    Recalling the degree-$4$ polynomials $h$ from \Cref{rem:richardson}, which exactly saturates the bound $\Opt_{\pm}(\p) \leq \sqrt{\binom{\N/2}{2}}$ from \Cref{cor:lovasz4}, it is easy to see that the optimum is achieved by a Gaussian state, $\rho = (1 + \i \X_1 \X_2)(1 + \i \X_3 \X_4) \cdots (1 + \i \X_{\N-1} \X_\N)$.
\end{example}

\subsection{Prior work on optimization and Gaussian states}
Bravyi, Gosset, K\"{o}nig, and Temme~\cite{BGKT19} investigated the task of finding the optimal Gaussian state for a given homogeneous degree-$4$ Hamiltonian with real coefficients.\footnote{In fact their work also applied to polynomials with a mix of degree-$2$ and degree-$4$ terms; but, they showed a reduction to the homogeneous degree-$4$ case, and we stick to this for simplicity.}
Let us define:
\begin{definition}
    Given self-adjoint $\p \in \CC(K_\N)$, we define
    \[
        \Optgauss(\p) = \max \{ {\E}_\rho[\p] : \rho \text{ is a Gaussian state}\} = \max_{\substack{\Sig \in \R^{\N \times \N} \\ \text{antisymmetric} \\ \|\Sig\|_{\mathrm{op}} \leq 1}} \{{\E}_{\rho_{\Sig}}[\p]\}.
    \]
\end{definition}
For $\p$ degree-$4$ homogeneous with real coefficients, the work~\cite{BGKT19} introduced a polynomial-size SDP relaxation of $\Optgauss(\p)$,
\[
    \SDPgauss(\p) \geq \Optgauss(\p),
\]
which we recall in \Cref{def:SDPgauss} (in \Cref{sec:CW}).
Concerning it, they proved several results, including the following:
\begin{theorem}                                     \label{thm:sdpapx}
    (\cite{BGKT19,So09}.) Applying the SDP-rounding algorithm of So~\cite{So09} gives a randomized polynomial-time rounding algorithm for $\SDPgauss(\p)$ that with high probability finds a Gaussian state achieving at least $\frac{1}{O(\log \N)} \SDPgauss(\p)$.
\end{theorem}
We present and generalize the above result in \Cref{sec:CW}.
\begin{theorem}                                     \label{thm:sdpweak}
    (\cite{BGKT19}.)
    $\displaystyle \SDPgauss(\p) \geq \tfrac{1}{\N} \Opt(\p)$ for all degree-$4$ homogeneous $\p \in \CC(K_\N)$.
\end{theorem}
\begin{remark} \label{rem:inspect}
    Inspecting the proof of \Cref{thm:sdpweak}, one sees it implies the stronger result $\SDPgauss(\p) \geq \frac{1}{\N} \SOS_4(\p)$.
\end{remark}
We also note  the following result shown by Haldar, Tavakol, and Scaffidi~\cite{HTS21}:%
\footnote{The eigenvalue bound for random Gaussian matrices that they rely on is standard.}
\begin{theorem}                                     \label{thm:gaussian-bad}
    (\cite{HTS21}.) For $\boldp \sim \SYK_4(\N)$, with high probability~$\Optgauss(\boldp) \leq O(1)$.
\end{theorem}
\begin{remark} \label{rem:inspect2}
    Inspecting the proof of \Cref{thm:gaussian-bad}, one sees it implies the stronger result $\SDPgauss(\boldp) \leq O(1)$, with high probability.
\end{remark}
\begin{remark}
One may also prove:
for $\boldp \sim \SYK_4(\N)$, with high probability~$\Optgauss(\boldp) = \Omega(1)$.
We define two $\N$-by-$\N$ matrices.  Let $g_0$ be the block-diagonal matrix with the first $\N/4$ blocks being $\begin{pmatrix} 0 & i \\ -i & 0 \end{pmatrix}$, and the remaining $\N/4$ blocks vanishing.
Let $g_1$ be the matrix defined as follows: the $k,l$ entry of $j$ is only nonvanishing if $k>\N/2$ and $l>\N/2$, in which case the entry is equal to
$\sum_{i<j\leq \N/2} (g_0)_{i,j} J_{\{i,j,k,l\}}$ for $k<l$ and has the opposite sign for $k>l$.
Finally, consider the covariance matrix $g_0 + Cg_1$, where $C$ is a scalar.
The entries of the matrix $g_1$ are independent\footnote{Up to the antisymmetry of the matrix, of course.} Gaussian
random variables, and the nonvanishing entries have variance $\Theta(1/\N^3)$, so standard random matrix theory results show that $\| g_1 \|_{\mathrm{op}} = O(1/\N)$.
Hence we may take $C=\Theta(\N)$ and have this covariance matrix define a Gaussian state.
One may verify that the expectation of $h$ for this covariance matrix is
$C \cdot \Theta(1/\N)$, so for $C=\Theta(\N)$ this is $\Theta(1)$.
\end{remark}
\begin{remark}
    As we show in \Cref{thm:fools4}, for $\boldp \sim \SYK_4(\N)$, with high probability $\SOS_4(\N) \geq \Omega(\N)$.
    In this case $\SDPgauss(\boldp) \geq \Omega(1)$ by \Cref{rem:inspect}, and so $\SDPgauss(\boldp) = \Theta(1)$ by \Cref{rem:inspect2}.
\end{remark}

Bravyi--Gosset--K\"{o}nig--Temme end their work~\cite{BGKT19} by asking whether for all~$\p$ one might have $\SDPgauss(\p) \geq \frac{1}{C} \Opt(\p)$, for some universal constant~$C$.
Our \Cref{thm:CW} below shows that this is true when $\Opt(\p)$ is within a constant factor of the upper bound from \Cref{cor:lovasz4} (namely $\N$ times the $2$-norm of $\p$'s coefficients).
Unfortunately, the inequality cannot be true in general: any such $C$ must be at least~$\Theta(\sqrt{\N})$.
This follows from \Cref{thm:gaussian-bad} together with the result (see our \Cref{thm:confirmsqrt}) that $\Opt(\boldp) \geq \Omega(\sqrt{\N})$ with high probability for $\boldp \sim \SYK_4(\N)$.
Whether the worst-case factor between $\Opt(\p)$ and $\Optgauss(\p)$ is more like~$\sqrt{\N}$ or more like~$\N$ (or something in between), we leave as an open question.

\subsection{Rounding $\SDPgauss(\p)$}   \label{sec:CW}
Let us now describe the SDP relaxation for degree-$4$ Hamiltonians introduced in~\cite{BGKT19}.
\begin{notation}
\label{vecdef}
    Given $A \in \C^{m \times m}$, let $\Vec(A) \in \C^{m} \otimes \C^m$ denote the vector formed by stacking the columns of~$A$, in order.
    Conversely, let $\Mat : \C^{m} \otimes \C^m \to \C^{m \times m}$ denote the inverse operation.
    We remark that
    \[
        \Tr_1(\Vec(A) \Vec(B)^\transp) = A^\transp B, \qquad \Tr_2(\Vec(A) \Vec(B)^\transp) = A B^\transp,
    \]
    where $\Tr_j$ denotes tracing out the $j$th tensor component of $\C^m \otimes \C^m$.
\end{notation}
\newcommand{\FF}{R}
\begin{definition}  \label{def:SDPgauss}
    Let $\Sig \in \R^{\N \times \N}$ be antisymmetric with $\|\Sig\|_{\mathrm{op}} \leq 1$.
    Define $\FF = \Vec(\Sig)\Vec(\Sig)^\transp$.
    Note that we have the following properties:
    \begin{enumerate}
        \item $\FF \geq 0$;
        \item $\Tr_1(\FF) \leq \Id$ (since $\Tr_1(\FF) = \Sig^\transp \Sig$ and $\|\Sig\|_{\mathrm{op}} \leq 1$);
        \item $\Tr_2(\FF) \leq \Id$ (similarly, since $\Tr_2(\FF) = \Sig \Sig^\transp$);
        \item $\braket{j_1 k_1|\FF|j_2 k_2}$ changes sign if $j_1,k_1$ are interchanged, or if $j_2,k_2$ are interchanged (since $\Sig$ is antisymmetric).
    \end{enumerate}
    Furthermore, if $\p = \sum_{S \in \binom{[\N]}{4}} a_S \X^S \in \CC(K_\N)$ (with $\N$~even), and we introduce the matrix
    \[
        H = -\sum_{S \in \binom{[\N]}{4}} a_S (\ketbra{S_1 S_2}{S_3 S_4} - \ketbra{S_1 S_3}{S_2 S_4} + \ketbra{S_1 S_4}{S_2 S_3}),
    \]
    then $\E_{\rho_{\Sig}}[\p] = \tr(\rho_{\Sig} \p) = \Tr(\FF H)$.
    We thereby define
    \[
        \SDPgauss(\p) = \max\{ \Tr(\FF H) : \FF \text{ satisfying properties 1--4 above}\},
    \]
    which is a semidefinite programming relaxation of $\Optgauss(p)$.
\end{definition}
Given this SDP, \cite{BGKT19} observe that the efficient randomized rounding algorithm of So~\cite{So09} establishes $\Optgauss(p) \geq \frac{1}{O(\log \N)} \SDPgauss(p)$.
So's work is a kind of ``noncommutative analogue'' of the $\frac{1}{O(\log \N)}$-factor approximation algorithm for quadratic programming due to Nestrov~\cite{Nes98}, Nemirosvsky--Roos--Terlaky~\cite{NRT99}, and Megretski~\cite{Meg01}.
In the ``commutative case'', said algorithm was generalized by Charikar and Wirth~\cite{CW04} to show that there is an $\frac{1}{O(\log(1/\eps))}$-factor approximation for Max-Cut whenever the optimum is within a factor of~$\eps$ of the ``naive upper bound'' (see \Cref{rem:naive}).
Here we are able to generalize the Charikar--Wirth result to our noncommutative setting:
\begin{theorem} \label{thm:CW}
    There is an efficient randomized algorithm that, given $\p = \sum_{S \in \binom{[\N]}{4}} a_S \X^S$ with real~$\vec a \in \R^{\binom{[\N]}{4}}$ satisfying $|\vec a| = 1$ and
    $
        \SDPgauss(p) = \eps \N,
    $
    with high probability outputs the covariance matrix~$\Sig$ of a Gaussian state satisfying
    \[
        {\E}_{\rho_\Sig}[\p] = \tr(\rho_\Sig \p) \geq \frac{\eps}{O(\log(1/\eps))} \N.
    \]
\end{theorem}
\begin{remark}
    From \Cref{cor:lovasz4} we know that $\eps \leq \frac{1}{2\sqrt{2}}$ always.
    It's also trivial to show $\eps \geq 1/\N^3$, say, meaning the $O(\log(1/\eps))$ factor is never worse than $O(\log \N)$, recovering \Cref{thm:sdpapx}.
    In fact, taking care in our proof shows that this approximation factor need not be worse than $(1+o(1))\sqrt{2 \ln \N}$.
\end{remark}
\begin{remark}
    It is straightforward to observe that our algorithm can be derandomized in $\N^{O(\log(1/\eps))}$ time, which is polynomial in the case $\eps = \Omega(1)$ of most interest.
    The algorithm employs $\N^2$ Rademacher random variables, and the analysis only needs $4$-wise uniformity for \Cref{ineq:e4} (see e.g.~\cite[Thm.~4.3]{AH11}), $2$-wise uniformity for \Cref{ineq:err1}, and $(p+2)$-wise uniformity for \Cref{ineq:err22}, with $p = O(\log(1/\eps))$.
    Thus the derandomization follows from~\cite{ABI86}.
\end{remark}
\begin{proof}
    By standard semidefinite programming methods~\cite{GLS81}, one can efficiently find a feasible solution~$\FF$ achieving~$\SDPgauss(p) = \eps \N$.\footnote{See the technicalities discussed in \Cref{sec:complexity}.}
    We now describe a randomized rounding algorithm that finds the required~$\Sig$ with constant probability.
    We may boost the success probability arbitrarily close to $1$ by repeating the algorithm constantly many times and choosing the best result achieved.

    Since $\FF \geq 0$ we may write a Cholesky factorization $\FF = L L^\transp$ which, by property~4 in \Cref{def:SDPgauss} will satisfy
    \begin{equation}    \label{eqn:antisym}
        \bra{k,j}L = -\bra{j,k}L \quad \forall j,k \in [\N].
    \end{equation}
    Now define
    \[
        \bSig = \sigma \Mat(L \vec \bx) = \sum_{j,k=1}^{\N} \bx_{jk} \sigma A^{(jk)}, \quad A^{(jk)} \coloneqq \Mat(L\ket{jk}),
    \]
    where $\sigma$ is a positive constant to be chosen later, and $\vec \bx$ is a vector of $\N^2$ independent Rademacher random variables.
    Note that each $A^{(jk)}$ is antisymmetric, by \Cref{eqn:antisym}; hence $\bSig$ is always antisymmetric and thus defines a Gaussian pseudostate.
    Further,
    \[
        \E[\Vec(\bSig) \Vec(\bSig)^\transp] = \sigma^2 L \E[\ketbra{\bx}{\bx}] L^\transp = \sigma^2 \FF,
    \]
    from which we see that
    \[
        \E[\tr(\rho_{\bSig} \p)] = \sigma^2 \Tr(\FF H) = \sigma^2 \SDPgauss(\p) = \sigma^2 \cdot \eps \N.
    \]
    Since $\tr(\rho_{\bSig} \p)$ is a quadratic polynomial in Rademachers, it follows (e.g.~\cite[Thm.~9.24]{OD14}) that
    \begin{equation}    \label[ineq]{ineq:e4}
        \tr(\rho_{\bSig} \p) \geq  \sigma^2 \cdot \eps \N \text{ with probability at least } e^{-4}/4 > .004.
    \end{equation}
    (Remark: So~\cite{So09} uses a more complicated construction for~$\bSig$ that ensures $\tr(\rho_{\bSig} \p) = \sigma^2 \SDPgauss(\p)$ holds with certainty.
    He then shows $\|\bSig\|_{\mathrm{op}} \leq 1$ with high probability with the choice $c \sim \frac{1}{\sqrt{2e \ln \N}}$.
    Our methods below would also show this even for $c \sim \frac{1}{\sqrt{2 \ln \N}}$.)

    We now need to study the eigenvalues of~$\bSig$.
    It will be mildly more convenient to study the eigenvalues of
    \[
        \bG \coloneqq \i\bSig = \sum_{j,k=1}^{\N} \bx_{jk} B^{(jk)}, \quad \text{where } B^{(jk)} \coloneqq \sigma \i A^{(jk)},
    \]
    since the matrices~$B^{(jk)}$ are self-adjoint.
    Note that
    \[
        \sum_{j,k=1}^{\N} {A^{(jk)}}^\transp A^{(jk)} = \sum_{j,k=1}^{\N} \Tr_1(L\ketbra{jk}{jk} L^\transp) = \Tr_1(\FF) \leq \Id
    \]
    (and similarly $\sum_{j,k=1}^{\N} A^{(jk)}{A^{(jk)}}^\transp  \leq \Id$), and this is equivalent to
    \begin{equation}    \label[ineq]{ineq:bI}
        \sum_{j,k=1}^{\N} (B^{(jk)})^2 \leq \sigma^2 \Id.
    \end{equation}
    At this point one can begin to analyze~$\bG$ using standard theory for sums of random self-adjoint matrices; e.g., $\Pr[\lambda_{\max}(\bG) \geq t] \leq \N e^{-t^2/(2\sigma^2)}$ holds~\cite{Tro12}.

    Say $\bG$ has (real) eigenvalues $\blambda_1, \dots, \blambda_\N$ and associated orthonormal eigenvectors $\vec \bphi_1, \dots, \vec \bphi_\N$.
    Our algorithm may form
    \[
        \widehat{\bG} = \sum\{ \blambda_j \ketbra{\bphi_j}{\bphi_j} : |\blambda_j| \leq 1\},
    \]
    which is also self-adjoint, and then $\widehat{\bSig} = \i \widehat{\bG}$.
    This $\widehat{\bSig}$ is the covariance matrix of a Gaussian state~$\rho_{\widehat{\bSig}}$, and by \Cref{fact:wick4} we have
    \begin{equation}    \label{eqn:errdef}
        \textbf{err} \coloneqq \abs*{\tr(\rho_{\bSig} \p) - \tr(\rho_{\widehat{\bSig}} \p)}
        = \abs*{\sum_{S \in \binom{[\N]}{4}} a_S(e_1(S) - e_2(S) + e_3(S))},
    \end{equation}
  where
    \[
        \bolde_1(S) = \bSig_{S_1 S_2} \bSig_{S_3 S_4} - \widehat{\bSig}_{S_1 S_2} \widehat{\bSig}_{S_3 S_4}
    \]
    and $\bolde_2(S), \bolde_3(S)$ are defined similarly (with subscripts $S_1S_3,S_2S_4$ and $S_1S_4,S_2S_3$, respectively).
    By Cauchy--Schwarz,
    \[
        \textbf{err}^2 \leq \parens*{\sum_S a_S^2} \parens*{\sum_S (\bolde_1(S) - \bolde_2(S) + \bolde_3(S))^2} \leq 3\sum_S \bolde_1(S)^2 + 3\sum_S \bolde_2(S)^2 + 3\sum_S \bolde_3(S)^2.
    \]
    We will bound the first sum on the right, with the other two being bounded similarly.
    Since
    \[
        (rs - \widehat{r}\widehat{s})^2 = (r(s - \widehat{s}) + \widehat{s}(r-\widehat{r}))^2 \leq 2r^2(s-\widehat{s})^2 + 2\widehat{s}^2(r-\widehat{r})^2,
    \]
    we get
    \[
        \sum_{S} \bolde_1(S)^2 \leq 2 \sum_S \bSig_{S_1S_2}^2(\bSig_{S_3S_4} - \widehat{\bSig}_{S_3S_4})^2 + 2 \sum_S {\widehat{\bSig}_{S_3S_4}}^2(\bSig_{S_1S_2} - \widehat{\bSig}_{S_1S_2})^2.
    \]
    Here, the first expression on the right is at most
    \[
        2\parens*{\sum_{1 \leq j < k \leq \N} \bSig^2_{jk}}\parens*{\sum_{1 \leq j < k \leq \N} (\bSig_{jk} - \widehat{\bSig}_{jk})^2} = \tfrac12 \|\bSig\|_2^2 \|\bSig - \widehat{\bSig}\|_2^2,
    \]
    where we used the antisymmetry of $\bSig, \wh{\bSig}$.
    We can use similar considerations to bound the second expression, and will also use the evident fact that $\|\widehat{\bSig}\|_2 \leq \|\bSig\|_2$.
    Putting everything together, we conclude
    \[
        \textbf{err}^2 \leq \tfrac92 \|\bSig\|^2_2 \|\bSig - \widehat{\bSig}\|^2_2 = \tfrac92 \|\bG\|^2_2 \|\bG - \widehat{\bG}\|^2_2.
    \]
    One easily calculates that $\E[\|\bG\|^2_2] = \Tr(\sum_{jk} (B^{(jk)})^2) \leq \sigma^2 \N$ (by \Cref{ineq:bI}), and therefore by Markov we have
    \begin{equation}    \label[ineq]{ineq:err1}
        \textbf{err}^2 \leq C \sigma^2\N \cdot \|\bG - \widehat{\bG}\|^2_2 \text{ except with probability at most } .001,
    \end{equation}
    where $C$ is a universal constant.

    To analyze this final expression we appeal to Khintchine's inequality with operator coefficients (originally due to Lust-Piquard~\cite{LP86}, improved in~\cite{LPP91}), which tells us that for all $p \geq 1$,
    \[
        \E[\|\bG\|^p_p] \leq (C' \sqrt{p})^p \norm*{\sqrt{\sum_{jk} (B^{(jk)})^2}}_p^{p} \leq (C' \sqrt{p} \sigma)^p \N,
    \]
    where $C'$ is another universal constant, and we used \Cref{ineq:bI} again.
    Note that
    \[
        \|\bG - \widehat{\bG}\|^2_2 = \sum_{j=1}^{\N} \blambda_j^2 \cdot \bone[|\blambda_j| \geq 1] \leq \sum_{j=1}^{\N} \blambda_j^2 \cdot |\blambda_j|^p = \|\bG\|_{p+2}^{p+2},
    \]
    regardless of~$p$.
    We choose $p$ so that $\sqrt{p+2} = (e C' \sigma)^{-1}$ (so we need to also choose $\sigma \leq (\sqrt{3} e C')^{-1}$ so that $p \geq 1$).
    With this choice, we get
    \[
        \E[\|\bG - \widehat{\bG}\|^2_2] \leq  \E[\|\bG\|_{p+2}^{p+2}] \leq \exp\parens*{-\frac{1}{C''\sigma^2}} \N.
    \]
    Employing Markov again, and combining with \Cref{ineq:err1}, we get
    \begin{equation}    \label[ineq]{ineq:err22}
        \textbf{err} \leq C_1 \sigma \exp\parens*{-\frac{1}{C_2\sigma^2}} \N \text{ except with probability at most } .002.
    \end{equation}
    Finally, from \Cref{ineq:e4,eqn:errdef} we have shown that the algorithm constructs Gaussian state $\rho_{\widehat{\bSig}}$ satisfying
    \[
        \tr(\rho_{\widehat{\bSig}} \p) \geq  \parens*{\sigma^2 \cdot \eps - C_1 \sigma \exp\parens*{-\frac{1}{C_2\sigma^2}}} \N
    \]
    with probability at least $.002$.
    The proof is completed by taking $\sigma = c/\sqrt{\log(1/\eps)}$ for a sufficiently small constant $c > 0$.
\end{proof}

\section{Quantum lower-bound witnessing for degree-$4$ SYK} \label{sec:variational}

The main content of this section is verifying the correct order of magnitude for the largest eigenvalue of the~$\SYK_4$ model:
\begin{theorem}     \label{thm:confirmsqrt}
    For $\boldp \sim \SYK_4(\N)$, we have $\Opt(\boldp) \geq \Omega(\sqrt{\N})$ with high probability.
\end{theorem}
\begin{remark}
    Recall the physics prediction is that $\Opt(\boldp)$ concentrates around $\frac{1}{2\sqrt{2}}\sqrt{\N}$.
    We have not attempted to optimize the implied constant in the $\Omega(\cdot)$ of our theorem.
\end{remark}
Having shown this theorem, it will follow straightforwardly from its proof that there is an efficient quantum \emph{witnessing} algorithm that establishes this lower bound.

Actually, it will be more convenient to work with a ``$2$-colored'' (or ``$1$-variable decoupled'') version of the $\SYK$ model, which we  define below; this $2$-colored version is a specific case of the models considered in \cite{gross2017generalization}.
In \Cref{sec:witnessing} we will see an easy way to extend our results for the $2$-colored variant to the standard $\SYK$ model.
\begin{definition}  \label{def:2col}
    We define the \emph{$2$-colored $\SYK_4$ model} as follows.
    Given even $\N_1, \N \in \NN$, we work in $\CC(K_{\N_1+\N})$, with the indeterminates named $\varphi_1, \dots, \varphi_{\N_1}$, $\X_1, \dots, \X_\N$.
    We write $\boldp \sim \SYKtwo_4(\N_1,\N)$ to denote that $\boldp$ is constructed as
    \begin{equation} \label{eqn:2colh}
        \boldp = \frac{\i}{\sqrt{\N}} \sum_{m = 1}^\N \boldetas_m \X_m,
    \end{equation}
    where each $\boldetas_m \sim \SYK_3(\varphi_1, \dots, \varphi_{\N_1})$ independently; i.e.,
    \begin{equation}  \label{eqn:2coleta}
        \boldetas_m = \frac{\i}{\sqrt{\binom{\N_1}{3}}} \sum_{\substack{S \subseteq [\N_1] \\ |S| = 3}} \bJ^{(m)}_S \varphi^{S},
    \end{equation}
    where each $\bJ^{(m)}_S$ is an independent standard Gaussian.
\end{definition}
\begin{remark}
    One can infer from the proof of \Cref{thm:witness-alg} that, for $\N$ a multiple of~$4$,
    \[
        \E_{\boldp \sim \SYK_4(\N)}[\Opt(\boldp)] \geq  \tfrac{3\sqrt{3}}{8} \E_{\boldp_1 \sim \SYKtwo_4(\frac34\N, \frac14\N)}[\Opt(\boldp_1)].
    \]
\end{remark}

\begin{remark}
    In the remainder of this section we may sometimes tacitly assume that $\N$ or $\N_1$ is a multiple or $2$ or~$4$ or~$8$.
    This is purely for convenience, and the reader may easily verify these assumptions are inessential.
\end{remark}

\subsection{The main variational argument}  \label{subsec:variational}
\newcommand{\sig}{\sigma}
\begin{theorem}     \label{thm:confirmsqrt2}
    For $\boldp \sim \SYKtwo_4(\N_1,\N)$ with $\N_1 = \Theta(\N)$ even, we have $\Opt(\boldp) \geq \Omega(\sqrt{\N})$ except with probability at most $\exp(-\Omega(\N))$.
\end{theorem}
\begin{proof}
    Initially we will work in slightly more generality, with a setup similar to that of \Cref{sec:refutingSYK4}.
    Specifically, we assume a self-adjoint polynomial
    \[
        \p = \frac{\i}{\sqrt{\N}} \sum_{m=1}^\N \etas_m \X_m
    \]
    where initially we only need to assume that all indeterminates $\etas_j$ and $\X_k$ are self-adjoint, and that each $\etas_j$ anticommutes with each~$\X_k$.
    We will be motivated by the idea that the $\etas_j$'s ``almost'' satisfy $\etas_j^2 = 1$ and $\etas_j \etas_k = - \etas_k \etas_j$.
    To this end\footnote{The construction here perhaps be regarded as a quantum version of the classical construction in \cite{hastings2019duality}: we start with some state (in that paper a random classical state, and here a Gaussian state), and apply some transformation (there a local update, here a unitary) to improve the expectation value of the objective in the state up to linear order in some parameter.  In both constructions, the problem is to control higher order corrections in that parameter, though this is done very differently in both cases.}, we artificially adjoin to our algebra additional self-adjoint indeterminates $\sig_1, \dots, \sig_\N$ that \emph{do} satisfy $\sig_j^2 = 1$ and $\sig_j \sig_k = -\sig_k \sig_j$, and that additionally anticommute with all $\X_j$'s and $\etas_j$'s.
    Now if
    \[
        \p' \coloneqq \frac{\i}{\sqrt{\N}} \sum_{m=1}^\N \sig_m \X_m,
    \]
    it would be easy to optimize~$\p'$; we could choose the Gaussian state
    \[
        \rho_0 \coloneqq (1 - \i \sig_1 \X_1) \cdots (1 - \i \sig_\N \X_\N)
    \]
    which achieves $\tr(\rho_0 \p') = \sqrt{\N}$ (similarly to \Cref{eg:richardson}).
    The idea is now to ``rotate'' the $\sig_j$'s toward the $\etas_j$'s by introducing the skew-adjoint
    \[
        \zeta \coloneqq \tau_1 \sig_1 + \cdots + \tau_\N \sig_\N
    \]
    and then defining (for $\theta \in \R$)
    \begin{equation}    \label{eqn:rhotheta}
        \rho_\theta \coloneqq \Ad_{e^{- \theta \zeta}}(\rho_0), \qquad \text{where } \Ad_U(g) = U g U^{-1}.
    \end{equation}
    (Here, to make sense of $e^{- \theta \zeta}$, recall that we will ultimately take the $\etas_j$'s to be polynomials in the additional anticommuting indeterminates $\varphi_1, \dots, \varphi_{\N_1}$, $\N_1$~even; then all indeterminates will be residing in the algebra $\CC(K_{\N_1 + 2\N})$.
    This algebra is isomorphic to $\C^{D \times D}$ for $D = \N_1/2 + \N$, and the matrix exponential is defined therein.)
    This $\rho_\theta$ is always a state, since $\i\zeta$ is self-adjoint and hence $e^{-\theta \zeta}$ is unitary.
    We now have
    \begin{equation}    \label{eqn:theval}
        \tr(\rho_\theta \p) = \tr(\Ad_{e^{-\theta \zeta}}(\rho_0) \p)  = \tr(\rho_0 \Ad_{e^{\i \theta \zeta}}(\p)) = \tr(\rho_0 \p_\theta), \quad \text{where }
        \p_\theta \coloneqq \Ad_{e^{ \theta \zeta}}(\p).
    \end{equation}
    We may now use the following standard Baker--Campbell--Hausdorff-type formula:
    \begin{proposition}                                     \label{prop:bch}
        For $[\cdot,\cdot]$ denoting the commutator, we have
        \[
             \p_\theta = \p + \theta [\zeta, \p] + \theta^2 \int_0^1 (1-\bs) \Ad_{e^{ \bs \theta \zeta}}([ \zeta, [ \zeta, \p]]) \, \mathrm{d}s.
        \]
    \end{proposition}
    \begin{proof}
        (Sketch.)
        This formula is Taylor-expanding $A(\theta) = \Ad_{e^{ \theta \zeta}}$ around $\theta = 0$, using that $A'(\theta) = a A(\theta)$ for $a(\cdot) = [\zeta, \cdot]$ (see, e.g.,~\cite[Ch.~1.2]{Ros02}), and also using that $[\zeta, A(\theta)(g)] = A(\theta)([\zeta, g])$ since $A(\theta)$ and $\zeta$ commute.
    \end{proof}
    Putting \Cref{prop:bch} into \Cref{eqn:theval} and using $\tr(\rho_0 \p) = 0$, we infer
    \begin{align}
        \tr(\rho_\theta \p) &= \theta \tr(\rho_0 [\zeta, \p]) + \theta^2 \E_{\bs \sim [0,1]} \Bigl[(1-\bs)\tr(\rho_{\bs\theta} [ \zeta, [ \zeta, \p]])\Bigr] \nonumber\\
        &\geq \theta \tr(\rho_0 [\zeta, \p]) - \theta^2 \Opt_{\pm}([ \zeta, [ \zeta, \p]]), \label[ineq]{ineq:into}
    \end{align}
    where we used the triangle inequality.

    Now is the point at which we will specialize to the $2$-colored SYK situation from the theorem statement (and \Cref{def:2col}), with $\p$ replaced by $\boldp$ as in \Cref{eqn:2colh}, and with the $\boldetas_j$'s being drawn from~$\SYK_3(\N_1)$ as in \Cref{eqn:2coleta}.
    (We will also write $\brho_\theta$ in place of $\rho_\theta$, as it has a dependence on the random $\boldetas_j$'s.)
    We have
    \[
        [\zeta, \boldp] = \frac{\i}{\sqrt{\N}} \sum_{j,k=1}^{\N} [\boldetas_j \sig_j, \boldetas_k \X_k],
    \]
    and one can see that only the $j=k$ terms survive when $\tr(\rho_0 \cdot)$ is applied.
    We deduce
    \[
        \tr(\rho_0 [\zeta, \boldp]) = \frac{1}{\sqrt{\N}} \sum_{k=1}^\N \tr(\{\boldetas_k, \boldetas_k\}),
    \]
    where $\{\cdot, \cdot\}$ denotes the anticommutator.
    It is not hard to compute that
    \[
         \tr(\{\boldetas_k, \boldetas_k\}) = \frac{2}{\binom{\N_1}{3}} \sum_{\substack{S \subseteq [\N_1] \\ |S| = 3}} (\bJ^{(k)})^2,
    \]
    and hence
    \begin{equation}    \label{eqn:putme}
        \E[\tr(\rho_0 [\zeta, \boldp])] = 2\sqrt{\N}, \qquad \Pr[\tr(\rho_0 [\zeta, \boldp]) < \sqrt{\N}] \leq \exp(-\Omega(\N_1^3)),
    \end{equation}
    the inequality here being a standard tail-bound for chi-squared random variables.

    To complete the analysis of \Cref{ineq:into}, we must now estimate  $\Opt_{\pm}([\zeta, [\zeta, \boldp]])$.
    As this involves a lot of bookkeeping, we will stop keeping track of constants and henceforth use a softer approach.
    Recall that we are assuming $C_1 \N \leq \N_1 \leq C_2 \N_2$ for some constants $C_1, C_2 > 0$.
    We henceforth use $C_3, C_4, \dots$ to denote either unspecified positive constants, or else quantities like $\binom{\N}{3}/\N^3$ which --- while technically dependent on~$\N$ --- are nevertheless bounded above and below by universal positive constants.

    For the sake of symmetry, we now introduce new independent mean-zero Gaussians $\bG_{rsta}$, for $r,s,t \in [\N_1]$, $a \in [\N]$, where $\bG_{rsta}$ has variance~$1$ if $r,s,t$ are all distinct, and otherwise $\bG_{rsta}$  has variance~$0$.
    Then  (using the fact that the sum and difference of independent Gaussians is Gaussian) it is equivalent in distribution to reexpress each $\boldetas_a$ as
    \[
        \boldetas_a = \frac{C_3\i}{\N^{3/2}} \sum_{r,s,t} \bG_{rsta} \varphi_r\varphi_s\varphi_t.
    \]
    Now
    \[
        [\zeta, [\zeta, \boldp]] = \frac{C_4}{\N^5} \sum_{r,s,t,u,v,w,x,y,z,a,b,d}
                                                                                            [\bG_{rsta} \varphi_r \varphi_s \varphi_t \sig_a,
                                                                                            [\bG_{uvwb} \varphi_u \varphi_v \varphi_w \sig_b,
                                                                                            \bG_{xyzd} \varphi_x \varphi_y \varphi_z \X_d]].
    \]
    Considering the inner commutator, we have (using \Cref{fact:ac}) that $\bG_{uvwb} \varphi_u \varphi_v \varphi_w \sig_b$ and $\bG_{xyzd} \varphi_x \varphi_y \varphi_z \X_d$ anticommute iff they share an odd number of indeterminates.
    In particular this inner commutator is nonzero only if $|\{u,v,w\} \cap \{x,y,z\}| \geq 1$.
    Furthermore, such a surviving commutator only anticommutes with $\bG_{rsta} \varphi_r \varphi_s \varphi_t \sig_a$ provided there is another indeterminate in common, meaning $|\{r,s,t\} \cap (\{u,v,w\} \symdiff \{x,y,z\})| + 1[a=b] \geq 1$.
    Thus
    \[
        [\zeta, [\zeta, \boldp]] = \frac{C_5}{\N^5} \sum_{r,s,\dots,d} f(r,s,\dots,d) \bG_{rsta}\bG_{uvwb}\bG_{xyzd} \varphi_r \varphi_s \varphi_t \sig_a \varphi_u \varphi_v \varphi_w \sig_b \varphi_x \varphi_y \varphi_z \X_d
    \]
    for some $0$-$1$ indicator~$f$, where $f$ depends only on certain equalities and non-equalities among its arguments and has the property
    \begin{equation}    \label{eqn:fenforce}
        f(r,s,\dots,d)  = 0 \text{ unless } |\{u,v,w\} \cap \{x,y,z\}| \geq 1 \text{ and } (|\{r,s,t\} \cap (\{u,v,w\} \symdiff \{x,y,z\})| + 1[a=b] \geq 1).
    \end{equation}
    We are now in a position to apply \Cref{thm:generic-upper}, together with \Cref{rem:seq,rem:better-g}; these imply
    \begin{equation}    \label[ineq]{ineq:put0}
        \E[\Opt_{\pm}([\zeta, [\zeta, \boldp]]) \leq C_0 \sqrt{\N}, \qquad \Pr[\Opt_{\pm}([\zeta, [\zeta, \boldp]]) \leq \tfrac12 C_0 \sqrt{\N}] \leq \exp(-\Omega(\N))
    \end{equation}
    (for some constant $C_0 < \infty$)
    provided we can prove via the $k$th-moment method ($k = \Theta(\N)$) that
    \begin{equation}    \label[ineq]{ineq:esta}
        \Pr\bracks*{\abs*{P(\bG) \geq C_0\sqrt{n}}} \leq \exp(-C_6 \N), \qquad P(\bG) \coloneqq \frac{C_5}{\N^5} \sum_{r,s,\dots,d} f(r,s,\dots,d) \bG_{rsta}\bG_{uvwb}\bG_{xyzd}.
    \end{equation}
    (Here $C_0$ may depend on $C_1, \dots, C_6$).
    Once we have this, we can put \Cref{ineq:put0} and \Cref{eqn:putme} into \Cref{ineq:into}, and then
    \begin{equation}    \label{eqn:final}
        \theta = \frac{1}{C_0} \quad\implies\quad \E[\tr(\brho_\theta \boldp)] \geq \frac{1}{C_0} \sqrt{\N}, \qquad \Pr\bracks*{\tr(\brho_\theta \boldp) \leq \frac{1}{2C_0} \sqrt{\N}} \leq \exp(-\Omega(\N)),
    \end{equation}
    as needed.

    The remainder of the proof is therefore devoted to establishing \Cref{ineq:esta}, by the $k$th-moment method.
    (We will in fact choose $k = \N$, which we assume is an even integer.)
    Notice that in the sum defining~$P$, the presence of~$f$ forces at least two indices to match; hence this sum effectively has only $O(\N^{10})$ terms, not~$O(\N^{12})$.
    Thus $\E[P(\bG)^2] \leq O(1)$, and so $\Pr[|P(\bG)| \geq C_0 \sqrt{\N}] \leq \exp(-\Theta(\N))$ is plausible --- but only if $P(\bG)$ has Gaussian-like tails, despite being a degree-$3$ polynomial in Gaussians.

    Let us now make some further simplifications.
    First, recall that $\bG_{rsta}$ is defined to be~$0$ if $r, s, t$ are not all distinct.
    We may drop this assumption, and assume that \emph{every} $\bG_{rsta}$ is a standard Gaussian, by redefining $f(r, s, \dots, d)$ so that \emph{it} is~$0$ unless $|\{r,s,t\}| = |\{u,v,w\}| = |\{x,y,z\}| = 3$.
    Next, this $f(r, s, \dots, d)$ is some Boolean function of all possible Kronecker delta functions $\delta_{jk}$ and their negations $(1-\delta_{jk})$, for $j,k \in \{r, s, \dots, d\}$.
    By using inclusion-exclusion, we may therefore write
    \[
        P(\bG) = \sum_{\ell = 1}^{C_7} \pm P_\ell(\bG), \qquad P_\ell(\bG) \coloneqq \frac{C_5}{\N^5} \sum_{r,s,\dots,d} f_\ell(r,s,\dots,d) \bG_{rsta}\bG_{uvwb}\bG_{xyzd},
    \]
    where each $f_\ell$ is a product only of Kronecker deltas~$\delta_{jk}$ (and \emph{not} their negations).
    Furthermore, because each~$f_\ell$ must enforce \Cref{eqn:fenforce}, it must be \emph{viable}, by which me mean it must include at least $\delta_{jk}\delta_{\ell m}$ for indices $j \in \{r,s,t,a\}$, $k, \ell \in \{u,v,w,b\}$ \emph{distinct}, and $m \in \{x,y,z,d\}$.
    By a union bound, and adjusting constants, it therefore suffices to achieve
    \[
        \Pr\bracks*{\abs*{P_\ell(\bG)} \geq C_0\sqrt{\N}} \leq \exp(-C_8 \N)
    \]
    by the $\N$th-moment method, for all viable~$P_\ell$.
    Given that we will use the moment method, observe that if the summands in some $P_\ell(\bG)$ are a subset of those in some other~$P_{\ell'}(\bG)$, it suffices for us to handle $P_{\ell'}(\bG)$; this is because the latter's even moments are only larger (using the fact that Gaussian monomials have nonnegative expectation).
    Thus it suffices for us to handle the viable $P_\ell(\bG)$'s that include a \emph{minimal} number of Kronecker deltas, namely~$2$.
    As we now have symmetry between the indices in our $\bG$'s, it therefore remains to handle one particular minimal case; say, $v = t$ and $x = u$, meaning
    \[
        P_0(\bG) \coloneqq \frac{C_5}{\N^5} \sum_{r,s,t,u,w,y,z,a,b,d} \bG_{rsta}\bG_{utwb}\bG_{uyzd}.
    \]
    To obtain the required tail bound of the form $\Pr[|P_0(\bG)| \geq C_0\sqrt{\N}] \leq \exp(-C_8 \N)$, it remains to establish the following inequality:
    \begin{equation}    \label[ineq]{ineq:finalshow}
        \|P_0(\bG)\|_{\N} \leq O(\sqrt{\N}).
    \end{equation}
    To show this we would like to use the decoupling aspect of \Cref{thm:latala}; however we don't yet have the property ``$A_{j_1 \cdots j_q} = 0$ when $j_1, \dots, j_q$ are not all distinct''.
    Thus we separate out the various equality/nonequality cases for $(r,s,t,a), (u,t,w,b), (u,y,z,d)$, writing
    \begin{equation}    \label{eqn:handle}
        P_0(\bG) =     R(\bG) + Q_1(\bG) + Q_2(\bG) + Q_3(\bG) + P'_0(\bG),
    \end{equation}
    where
    \begin{gather*}
        R(\bG) \coloneqq \frac{C_5}{\N^5} \sum_{r,s,t,a} \bG_{rsta}^3, \qquad 
        Q_1(\bG) \coloneqq \frac{C_5}{\N^5} \sum_{\substack{r,s,t,u,w,a,b \\ (r,s,t,a), (u,t,w,b), \text{ distinct}}} \bG_{rsta}\bG_{utwb}^2, \\
        Q_2, Q_3 \text{ defined similarly}, \qquad P'_0(\bG) \coloneqq \frac{C_5}{\N^5} \sum_{\substack{r,s,t,u,w,y,z,a,b,d \\ (r,s,t,a), (u,t,w,b), (u,y,z,d) \text{ distinct}}} \bG_{rsta}\bG_{utwb}\bG_{uyzd}.
    \end{gather*}  
    Bounding $R$ is very easy; we can use, say, hypercontractivity~\cite[Thm.~5.10]{Jan97} to deduce
    \[
        \|R(\bG)\|_n \leq \sqrt{n-1}^{3} \|R(\bG)\|_2 \leq O(n^{-3.5}) \cdot \norm*{\sum_{r,s,t,a} \bG_{rsta}^3}_2 = O(n^{-3.5})  \cdot \sqrt{O(n^4)} = O(n^{-1.5}) \ll O(\sqrt{n}).
    \]
    As for $Q_1$ (and similarly $Q_2, Q_3$), we can first drop the condition ``$(r,s,t,a),(u,t,w,b)$ distinct'' in its definition for simplicity, since (as previously noted) additional terms can only make even moments larger.  Then, using the decoupling aspect of \Cref{thm:latala}, it suffices to obtain the bound
    \begin{equation}    \label[ineq]{ineq:getme}
        \norm*{\sum_{r,s,t,a,u,w,b} \bG'_{rsta}\bG_{utwb}^2}_n \leq O(n^{5.5}),
   \end{equation}
   where $(\bG')_{rsta}$ are new independent standard Gaussians.  
   This is equivalent to 
   \[
        \norm*{\sum_{t=1}^n \bG''_t \boldsymbol{q}_t^2}_n^n = \E_{\boldsymbol{q}_t} \E_{\bG''_t}\bracks*{\parens*{\sum_{t=1}^n \bG''_t \boldsymbol{q}_t}^n} \leq  O(n^4)^n,
   \]
   where $(\bG'')_t$ are independent standard Gaussians, and $(\boldsymbol{q}_t)_t$ are independent chi-quared random variables with~$n^3$ degrees of freedom.  
   For fixed outcomes $\boldsymbol{q}_t = \chi_t$, the random variable $\sum_{t} \bG''_t \chi_t$ is a mean-zero Gaussian with variance $\nu^2 \coloneqq \sum_t \chi_t^2$, and hence its $n$th moment is bounded by $\nu^n \cdot \sqrt{n}^n$.
   Thus to verify \Cref{ineq:getme} it suffices to show
   \begin{equation} \label[ineq]{ineq:confirm}
        \E[\boldsymbol{\nu}^n]  \leq O(n^{3.5})^n.
   \end{equation}
    Now we use the following well known tail bound for the $\boldsymbol{q}_t$'s (see, e.g.~\cite[Lem.~1]{LM00}):
    \[
        \Pr[\boldsymbol{q}_t \geq (4c+1) n^3] \leq \exp(-cn^3) \quad \forall c \geq 1.
    \]
    A union bound now easily implies 
    \[
        \Pr[\boldsymbol{\nu} \geq 5cn^{3.5}] \leq \exp(-\tfrac12c n^3) \quad \forall c \geq 1 \quad\implies\quad \Pr[\boldsymbol{\nu'} \geq c] \leq \exp(-\tfrac12n^3)^c, \quad \boldsymbol{\nu'} \coloneqq \boldsymbol{\nu}/(5 n^{3.5}),
    \]
    and from this one easily confirms \Cref{ineq:confirm} and hence \Cref{ineq:getme}:
    \[
        \E[(\boldsymbol{\nu}')^n] \leq 1 + \sum_{j \in \NN} 2^{(j+1) n} \cdot \Pr[2^j \leq \boldsymbol{\nu}' \leq 2^{j+1}] \leq 1 + \sum_{j \in \NN} 2^{(j+1)n} \cdot \exp(-\tfrac12 n^3)^{2^j} \leq O(1).
    \]

   Finally, to complete the proof of \Cref{ineq:finalshow}, it remains to handle~$P'_0$ from \Cref{eqn:handle}; i.e., to bound $\|P'_0(\bG)\|_n \leq O(\sqrt{n})$.
    Using the decoupling aspect of \Cref{thm:latala}, it in turn suffices to bound the $n$th norm of
    \begin{align*}
        P_0(\bG^{(1)}, \bG^{(2)}, \bG^{(3)}) &= \frac{C_5}{\N^5} \sum_{r,s,t,u,w,y,z,a,b,d} \bG^{(1)}_{rsta}\bG^{(2)}_{utwb}\bG^{(3)}_{uyzd} \\
        &= \frac{C_9}{\N} \sum_{t,u} \bH_t \bK_{tu} \bL_u,
    \end{align*}
    where $(\bH_t)_t$, $(\bK_{tu})_{t,u}$, $(\bL_u)_u$ (for $t,u \in [\N]$) are new independent standard Gaussians.
    \Cref{thm:latala} now implies
    \[
        \|P'_0(\bG)\|_{\N} \leq \frac{C_{10}}{\N}\parens*{\sqrt{\N} \|A\|_{\{1,2,3\}} + \N\|A\|_{\{1,2\},\{3\}} + \N\|A\|_{\{1,3\},\{2\}} + \N\|A\|_{\{2,3\},\{1\}} + \N^{3/2}\|A\|_{\{1\},\{2\},\{3\}}},
    \]
    where~$A \in \{0,1\}^{\N \times \N^2 \times \N}$ has $A_{h,jk,\ell} = 1$ iff $h = j$ and $k = \ell$.
    To obtain the desired tail bound of the form $\Pr[|P_0(\bG)| \geq C_0\sqrt{\N}] \leq \exp(-C_8 \N)$, it suffices to show $\|P_0\|_{\N} \leq O(\sqrt{\N})$.
    But this is true, as one can easily compute
    \begin{gather*}
         \|A\|_{\{1,2,3\}} = \N, \qquad \|A\|_{\{1,2\},\{3\}} = \|A\|_{\{2,3\},\{1\}} = \|\Id_{\N \times \N} \otimes \underbrace{(1, \dots, 1)}_{\text{$\N$}}\|_{\text{op}} = \sqrt{\N}, \\
         \|A\|_{\{1,3\},\{2\}} = \|\Id_{\N^2 \times \N^2}\|_{\text{op}} = 1,
    \end{gather*}
    and
    \[
        \|A\|_{\{1\},\{2\},\{3\}} = \max_{\substack{x,z \in \R^{\N}, y \in \R^{\N^2} \\ \|x\|,  \|y\|, \|z\| \leq 1}} \sum_{t,u} x_t y_{tu} z_u \leq \max_{x,y,z} \sqrt{\sum_{t,u} x_t^2 z_u^2}\sqrt{\sum_{t,u} y_{tu}^2} = 1,
    \]
    using Cauchy--Schwarz.
\end{proof}

\subsection{Efficient quantum lower-bound witnessing} \label{sec:witnessing}
The argument based on variational states used to prove \Cref{thm:confirmsqrt2}
can be readily turned into an efficient quantum lower-bound witnessing algorithm of the type discussed in \Cref{sec:random}.
(In the remainder of this section we will identify elements of~$\CC(\Gamma)$, $\Gamma = K_N$, with their representations in $\C^{N/2 \times N/2}$ under $\pi_\Gamma$.)
\begin{theorem} \label{thm:witness-alg}
    For a universal constant $c_0 > 0$, there is a $\poly(\N)$-time, $O(\N)$-space quantum algorithm~$W$ that, on input a degree-$4$ homogeneous $\p \in \CC(K_\N)$, outputs a quantum state $\rho \in \C^{D \times D}$ (with $D = 2^{\N/2}$ as in \Cref{prop:reps}).
    Except with probability $\exp(-\Omega(\N))$ over the draw of $\boldp \sim \SYK_4(\N)$, it holds for $\brho = W(\boldp)$ that $\tr(\brho \boldp) \geq c_0 \sqrt{\N}$.
\end{theorem}
\begin{remark}
    As described in \Cref{sec:random}, there is a $\poly(\N)\log(1/\delta)$-time randomized quantum procedure that, given \emph{any} worst-case $\p$ and corresponding $\rho = W(\p)$, \emph{certifies} that $\tr(\rho \p) \geq c_0\sqrt{\N}$ (up to any additive $1/\poly(\N)$) whenever it holds, except with probability at most~$\delta$.
\end{remark}
To prove \Cref{thm:witness-alg}, we first prove \Cref{prop:2col1} below, which establishes it for the $2$-colored $\SYK$ model analyzed in \Cref{thm:confirmsqrt2}.
Afterwards, extending to the usual $\SYK$ model is not difficult.
\begin{proposition} \label{prop:2col1}
    \Cref{thm:witness-alg} holds with $\boldp \sim \SYKtwo_4(\frac34\N, \frac14\N)$ in place of $\boldp \sim \SYK_4(\N)$.
\end{proposition}
\begin{proof}
    The input $\p$ to the algorithm resides in an algebra with $\frac34\N + \frac14\N$ anticommuting indeterminates, which we will simply call $\X_1, \dots, \X_\N$ (rather than $\varphi_1, \dots, \varphi_{\frac34 \N}, \X_1, \dots, \X_{\frac14\N}$).
    These are represented by the $\gamma$-matrices from \Cref{eqn:gamma}, over an $\N/2$-qubit Hilbert space.

    Now the target variational state~$\rho_\theta$  defined in \Cref{eqn:rhotheta} (with $\theta$ chosen as in \Cref{eqn:final}) in fact resides in $\CC(\Gamma)$ for $\Gamma = K_{\N + \frac14\N}$, since it involves the additional indeterminates $\sig_1, \dots, \sig_{\N/4}$.
    To accommodate this the algorithm will actually work in the enlarged Hilbert space with~$\N/8$ additional qubits.
    But since the input~$\p$ does not depend on the $\sig_j$'s, it's easy to see that if the algorithm manages to construct the quantum state~$\rho_\theta$, it may as its final step simply discard (trace out) its last~$\N/8$ qubits, leaving a quantum state~$\rho$ that has $\tr(\rho \p) = \tr(\rho_\theta \p)$.

    Thus we only need to show the algorithm can construct~$\rho_\theta$.
    Indeed, by slightly adjusting the constant~$c_0$ in the statement of \Cref{thm:witness-alg}, it suffices for it to construct $\rho_\theta$ up to sufficiently small~$1/\poly(\N)$ error in trace distance.
    The first step of the algorithm is to construct the Gaussian state~$\rho_0$.
    This may be done very efficiently; this state is maximally mixed on the first~$3\N/8$ qubits, and on the remaining~$\N/4$ qubits it is a stabilizer state\cite{gottesman1997stabilizer} with a simple description.
    (Indeed, had the last $\N/2$ indeterminates been ordered as $\X_{3\N/4}, \sig_1, \X_{3\N/4+1}, \sig_2, \dots$, then $\rho_0$ would simply be $2^{\N/8} \ket{00\cdots 0}\!\bra{00\cdots 0}$ on these qubits, as in \Cref{eg:ket0}.)

    It remains for the algorithm to rotate this state by $\exp(-\theta \zeta)$ (up to polynomially small error).
    This rotation can be accomplished efficiently: it is the same as Hamiltonian evolution under the Hamiltonian $-\i\zeta$ for time $\theta$.
    This Hamiltonian evolution can be done in polynomial time to any desired inverse polynomial error, for example by Trotter--Suzuki evolution, since the ``Hamiltonian" $-i\zeta$ is the sum of polynomially many terms (namely, the monomials in the Majorana operators), and unitary evolution under a given term can be done efficiently as each term is a product of Pauli operators.
    More efficient algorithms exist which can provide a polynomial speedup in $\N$ and exponential speedup in error~\cite{babbush2019quantum}.
\end{proof}

It now remains to move from the $2$-colored $\SYK$ model to the usual one.
For this, let us partition $\binom{[\N]}{4} = A_{\text{in}} \sqcup A_{\text{out}}$, where
\[
    A_{\text{in}} = \{\{j_1, j_2, j_3, j_4\} : j_1 < j_2 < j_3 \leq \tfrac34 \N < j_4\}, \quad A_{\text{out}} = \binom{[\N]}{4} \setminus A_{\text{in}},
\]
and let us also write any homogeneous degree-$4$ polynomial $\p$ as $\p_{\text{in}} + \p_{\text{out}}$, where $\p_{\text{in}}$ (respectively, $\p_{\text{out}}$) contains the monomials of~$\p$ corresponding to indices $S \in A_{\text{in}}$ (respectively, $A_{\text{out}}$).
It is then easy to see that for $\boldp \sim \SYK_4(\N)$ it holds that $\boldp_{\text{in}}$  is distributed as $\SYKtwo_4(\frac34\N, \frac14\N)$ up to a scaling factor; specifically, we have
\begin{equation} \label{eqn:forc}
    c \cdot \boldp_{\text{in}} \sim \SYKtwo_4(\tfrac34\N, \tfrac14\N) \quad \text{for } c \coloneqq \sqrt{\frac{\binom{\N}{4}}{\binom{3\N/4}{3} (\N/4)}} \leq \frac{8}{3\sqrt{3}}.
\end{equation}
We may now establish \Cref{thm:witness-alg}.
\begin{proof}[\Cref{thm:witness-alg}]
    Let $W$ be the quantum algorithm that, on input $\p \in \CC(K_\N)$ homogeneous of degree~$4$, applies the algorithm from \Cref{prop:2col1} to $c \cdot \p_{\text{in}}$, where $c$ is as in \Cref{eqn:forc}.
    Suppose now that
    \[
        \boldp = \frac{1}{\sqrt{\binom{\N}{4}}} \sum_{\substack{S \subseteq [\N] \\ |S| = 4}} \bJ_S \X^S \sim \SYK_4(\N).
    \]
    Then as discussed, $c \cdot \boldp_{\text{in}}$ is distributed as $\SYKtwo_4(\tfrac34\N, \tfrac14\N)$, and hence the output~$\brho$ of \Cref{prop:2col1}'s algorithm satisfies
    \begin{equation}    \label[ineq]{ineq:fin1}
        \Pr_{\bJ_S : S \in A_{\text{in}}}[\tr(\brho \boldp_{\text{in}}) \geq \tfrac{c'_0}{c} \sqrt{\N}] \geq 1 - \exp(-\Omega(\N))
    \end{equation}
    (for some constant $c'_0 > 0$).
    On the other hand, for \emph{any} fixed outcome $\brho = \rho$ we have
    \[
        \E_{\bJ_S : S \in A_{\text{out}}}[\tr(\rho \boldp_{\text{out}})] =  \E\bracks*{\frac{1}{\sqrt{\binom{\N}{4}}} \sum_{S \in A_{\text{out}}} \tr(\rho \X^S) \bJ_S} = 0.
    \]
    Moreover, since $|\tr(\rho \X^S)| \leq 1$ for each~$S$, Gaussian concentration implies that
    \begin{equation}    \label[ineq]{ineq:fin2}
        \Pr_{\bJ_S : S \in A_{\text{out}}}[|\tr(\rho \boldp_{\text{out}})| \geq \tfrac{c'_0}{2c} \sqrt{\N}] \leq \exp(-\Omega(\N)).
    \end{equation}
    Using
    \[
        \tr(\brho \boldp) = \tr(\brho \boldp_{\text{in}}) +  \tr(\brho \boldp_{\text{out}})
    \]
    and taking $c_0 = \tfrac{c'_0}{2c}$, \Cref{ineq:fin1,ineq:fin2} complete the proof.
\end{proof}

\section*{Open problems}
Here we briefly list some open problems not previously mentioned in the paper:
\begin{itemize}
    \item Can efficient certification algorithms, within a constant factor, be given for the SYK model with $q = 3$ or $q = 6$?
    \item For SYK of degree~$6$, can one show that constant-degree SOS is strongly fooled (i.e., degree-$k$ SOS fails to certify an $O(\sqrt{n})$ upper bound, for any constant~$k$)?
    \item Do our certification algorithms translate to the \emph{sparse} SYK model~\cite{XSYS20,GJRV21}, in which only $\Theta(n)$ coefficients are chosen to be nonzero?
    \item What can be said about the closely-related Erd\H{o}s--Schr\"{o}der model~\cite{ES14} of quantum spin glasses?
    \item For which deterministic SYK-like models can we prove SOS bounds?  For example, the eigenvalue bounds in \cite{klebanov2018spectra} can be proven within SOS.
    \item Can one give evidence that classical algorithms cannot efficiently certify $\Opt(\bp) \geq \Omega(\sqrt{\N})$ (with high probability) for $\bp \sim \SYK_4(\N)$?
    \item What can be said about optimization --- even of degree-$1$ Hamiltonians --- within $\CC(\bGamma)$ when $\bGamma \sim G(\N,p)$ is an Erd\H{o}s--R\'enyi random graph?
    \item What is the largest possible ratio between $\Opt(\p)$ and $\Optgauss(\p)$ for $\p \in \CC(K_\N)$ of degree~$4$?  We know it is at most $O(\N)$ and at least~$\Omega(\sqrt{\N})$ (the latter because of typical $\SYK_4(\N)$ instances).
    \item In the same way that  \Cref{thm:CW} extends~\cite{CW04}, can we extend~\cite{GW95} by showing that in the context of \Cref{thm:CW}, if $\Opt(\p) \geq (1-\delta)\sqrt{\binom{n/2}{2}}$ then an efficient algorithm can deliver a Gaussian state achieving expectation value at least~$(1-\delta')\sqrt{\binom{n/2}{2}}$, for some $\delta'$ with $\delta' \to 0$ as $\delta \to 0$ (e.g., $\delta' = O(\sqrt{\delta})$)?
\end{itemize}

\section*{Acknowledgments}
The second author would like to thank Pravesh Kothari, Sidhanth Mohanty, Tselil Schramm, and William Slofstra for their helpful comments.
The authors would also like to thank Maarten Stroeks for pointing out the need to handle $R, Q_1, Q_2, Q_3$ in \Cref{eqn:handle}.

\renewcommand{\i}{\origi}
\bibliographystyle{alpha}
\bibliography{quantum}
\appendix
\section{Derivatives of $\Opt(\ell)$ for $\ell$ supported on an independent set}
\label{localopt}
In \Cref{sec:lovaszth}, we have shown the bound
\[
        \alpha(\Gamma) \leq \OurPsi(\Gamma) \leq \th(\Gamma)
\]
for the $\OurPsi$ function.

Here we further consider the question of whether there is some graph $\Gamma$ such that
$\OurPsi(\Gamma)>\alpha(\Gamma)$.  As partial evidence that these two functions might coincide, we prove the following.
\begin{lemma}
Let $S$ be a maximal independent set in $\Gamma$.  Regard
$\Opt(\sum_j a_j \X_j)$ as a function of the vector $\vec a$ on the sphere $|\vec a|=1$.
Let $\vec a_0$ be the vector with $(\vec a_0)_j=0$ for $j\not \in S$ and
$(\vec a_0)_j=\frac{1}{\sqrt{|S|}}\lambda_j$ for $j\in S$, where the signs $\lambda_j\in \{\pm 1\}$
are chosen so that some common eigenvector of the $\X_j$
has associated eigenvalue $\lambda_j \in \{\pm 1\}$.  Hence, $\ell_0=\sum_j (\vec a_0)_j \X_j$ is the $\ell$ considered in \Cref{prop:alpha-lower}.

Then, the first derivative of this function with respect to $\vec a$ vanishes at $\vec a=\vec a_0$, and the Hessian matrix (i.e., matrix of second derivatives w.r.t. some coordinates on the sphere $|J|=1$) is negative semidefinite.

Remark: we say that the matrix is negative semidefinite rather than that we are at a local optimum as the matrix may have a zero eigenvalue.  Indeed, consider the case of two operators, $X,Z$ which anticommute.  Then $\ell=\cos(\theta) X + \sin(\theta)Z$ has largest eigenvalue equal to $1$ for all $\theta$.
\begin{proof}
A convenient set of coordinates on the sphere in a neighborhood of a given $\vec a=\vec a_0$ is to consider vectors $\rmd \vec a$ normal to $\vec a_0$ and let
$\vec a=(\vec a_0 + \rmd \vec a)/|\vec a_0+\rmd \vec a|$.
Let $$E=\Opt(\sum_\alpha ((\vec a_0)_j+\rmd \vec a_j) O_j).$$
Note that the argument of the optimum in $E$ is not normalized to have $\ell_2$ norm equal to $1$.
We have
$$\Opt\Bigl(\frac{\vec a_0 + \rmd \vec a}{|\vec a_0+\rmd \vec a|}\Bigr)=\frac{E}{|\vec a_0+\rmd \vec a|}.$$

The vanishing of the first derivative is immediate: operators $\X_k$ for $k\not \in S$ anticommute with at least one operator $\X_j$ for some $j\in S$ and so the expectation value of $\X_k$ vanishes in any eigenstate of $\ell_0$ of maximal eigenvalue.

After some calculus we find to second order, for $\rmd \vec a$ normal to $\vec a_0$, that
$$\frac{E}{|\vec a_0+\rmd \vec a|}=\sqrt{|S|}+\sum_{j,k} M_{j,k} \rmd \vec a_j \rmd \vec a_k-\sum_{j,k} \frac{E}{2} \delta_{j,k} \rmd \vec a_j \rmd \vec a_k$$
with $$M_{j,k}=\frac{1}{2}\partial_j \partial_j E,$$ where we write $\partial_j$ for brevity for $\partial_{\rmd \vec a_j}$.

So, we must show that the largest eigenvalue of $M$ in the given subspace (i.e., for vectors normal to $\vec a_0$) is bounded by $E/2=|\sqrt{S}|/2$.

We have, for $j,k \not \in S$ that
\be
\frac{1}{2}\partial_j \partial_k E = \langle \phi, \X_j (E-\ell_0)^{-1} \X_k \phi \rangle,
\ee
where $\phi$ is an eigenstate $\ell_0$ with maximal eigenvalue.
This follows from standard second-order perturbation theory results.

We claim that $M$ is block-diagonal.  Each block corresponds to a choice of $T\subset S$ and $k$ is in a given block iff $\X_k$ anticommutes with $\X_m$ for all $m\in T$ and commutes with $\X_m$ for all $m \in S\setminus T$.
Indeed, the state $\X_k \phi$
is an eigenstate of $\X_m$ for each $m\in S$ with eigenvalue $\pm 1$ (the sign is $-1$ if $m\in T$ and $+1$ otherwise), and $E-\ell_0$ commutes with all such $\X_l$, so
 $\partial_j \partial_k E$ vanishes between blocks.

Consider a given block for given $T$.  The operator $E-\ell_0$ equals $$ \frac{2|T|}{\sqrt{|S|}}$$ in that block.  Hence, the largest eigenvalue of $M$ in that block is bounded by $\frac{\sqrt{|S|}}{2|T|}$ times the cardinality of that block.  Hence it is less than $\sqrt{|S|}/2$ unless the cardinality of that block is larger than $|T|$.
However, in that case, the independent set was not maximal: we may remove all elements of $T$ from $S$ and add all elements of the given block to increase the independent set.
\end{proof}
\end{lemma}

\section{Low-rank and Gaussian states}
\label{lowrank}
Suppose
\be
p=\sum_{\alpha=1}^k \lambda_\alpha Q_\alpha^2,
\ee
where $\alpha$ is some index, $k$ is an integer called the {\it rank}, $\lambda_\alpha$
is a real scalar and
$Q_\alpha$ is a self-adjoint quadratic operator.

Let $Q_\alpha=\sum_{jk} \X_j (K_\alpha)_{jk} \X_k,$
where $K_{jk}$ is a self-adjoint anti-symmetric matrix (hence, $K$ is pure imaginary).
Let us fix
$|K_\alpha|_2=1$.  Hence, $\| Q_\alpha \|_{\mathrm{op}} = O(\N^{1/2})$.

Here we consider the low rank case, where $k=O(1)$.
Assume also for normalization that $\sum_\alpha |\lambda_\alpha|=O(1)$.

We have
$$[Q_\alpha,Q_\beta]=4\sum_{jk}  \Bigl([K_\alpha,K_\beta]\Bigr)_{jk} \X_j \X_k.$$
The $\ell_2$ norm of $[K_\alpha,K_\beta]$ is bounded (by a triangle inequality and Cauchy-Schwarz) by $2 |K_\alpha|_2 |K_\beta|_2\leq 2$, though probably tighter bounds are possible.
Hence, by results before, $\| [Q_\alpha,Q_\beta]\|_{\mathrm{op}}=O(\N^{1/2})$.

Similarly, \begin{align}\| [Q_\alpha,p] \|_{\mathrm{op}} & \leq
\sum_{\beta=1}^k |\lambda_\beta| \| [Q_\beta^2,Q_\alpha]\|_{\mathrm{op}} \nonumber \\ \nonumber &\leq
\sum_{\beta=1}^k |\lambda_\beta| \cdot O(\N).
\end{align}

Consider the re-scaled operators $\N^{-1/2} Q_\alpha$ and $\N^{-1} p$.  These operators all have operator norm $O(1)$ and the commutator of any two such operators is $O(\N^{-1/2})$.

By section VII of \cite{hastings2009making},  we can define a POVM (positive operator-valued measure) that approximately measures the $K$ different operators $\N^{-1/2} Q_\alpha$ and $\N^{-1} p$ on some arbitrary quantum state $\rho$.  In particular (dropping the re-scaling by $\N^{-1/2}$ and $\N^{-1}$ from here on), the measurement returns scalars $q_\alpha$ such that the average over outcomes of the expectation value $(Q_\alpha-q_\alpha)^2$ in the resulting state is $o(\N)$.
Further, for any $\alpha$, the expectation value $\tr(p \rho)$ is within $o(\N)$ of the average of $\sum \lambda_\alpha q_\alpha^2$ over measurement outcomes.
Let $\rho$ then be the state that maximizes $\tr(p \rho)$.  Apply this POVM.  The expectation value of $p$ in the resulting state (called $\rho_{outcome}$), averaged over outcomes, is within
$o(\N)$ of $\tr(p \rho)$.
Hence, the average over measurement outcomes of $\sum_\alpha \lambda_\alpha \tr(\rho_{outcome} Q_\alpha)^2$ is within $o(\N)$ of $\tr(p \rho)$.

So, there is some state $\sigma$ for which $\sum_\alpha \lambda_\alpha \tr(\sigma Q_\alpha)^2$ is within $o(\N)$ of $\tr(p \rho)$.
Let us optimize $\sum_\alpha \lambda_\alpha \tr(\sigma Q_\alpha)^2$ over states $\sigma$.
An argument with Lagrange multipliers shows that the maximum is attained on a Gaussian state.  Further, on this Gaussian state, one may verify that $\tr(p \sigma)$ is
within $o(\N)$ of $\sum_\alpha \lambda_\alpha \tr(\sigma Q_\alpha)^2$ and hence within $o(\N)$ of $\tr(p \rho)$.

\end{document}